\documentclass{lmcs} 


\usepackage{hyperref}
\usepackage{xspace,bm}
\usepackage{stmaryrd} 

\usepackage[pagewise]{lineno}

\usepackage{pstring}
\pstrSetArrowColor{black}

\usepackage{macros}

\def\eg{{\em e.g.}\xspace}
\def\ie{{\em i.e.}\xspace}

\begin{document}

\title[Reasoning about Strategies on Collapsible Pushdown Arenas]{Reasoning
  about Strategies on Collapsible Pushdown Arenas with
Imperfect Information}

\author[B.~Maubert]{Bastien Maubert}	
\address{Universit\`a degli Studi di Napoli Federico II, Naples, Italy}	

\author[A.~Murano]{Aniello Murano}	
\address{Universit\`a degli Studi di Napoli Federico II, Naples, Italy}	

\author[O.~Serre]{Olivier Serre}	
\address{Universit\'e de Paris, IRIF, CNRS, France}	





\begin{abstract}
    Strategy Logic with
  imperfect information (\SLi) is a very expressive logic designed to express
  complex properties of strategic abilities in distributed systems.
Previous work on $\SLi$ focused on \emph{finite} systems, and showed that the model-checking problem is
decidable when  information on the control states of the system is hierarchical
among the {players or components} of the system, meaning
that {the players or components} can be totally ordered
according to their respective knowledge {of the state}. We show that moving from finite to
infinite systems generated by {collapsible (higher-order)} pushdown systems  preserves decidability, under the natural restriction
that the stack content is visible.

The proof follows the same lines as in the case of finite systems, but
requires to use ({collapsible}) alternating pushdown tree automata.
Such automata are undecidable, but \emph{semi-alternating} pushdown tree automata were introduced
and proved decidable, to study a strategic problem on pushdown systems with two players.
In order to tackle multiple players with hierarchical information, we refine
further these automata: we define  \emph{direction-guided}
({collapsible}) pushdown tree automata, and show that they are
stable under {projection,
nondeterminisation and narrowing. For the latter operation, used to deal with
imperfect information, stability holds under some assumption that is
satisfied when used for systems with visible stack.} We then use these automata to prove our
main result.
\end{abstract}

\maketitle

\section{Introduction}
\label{sec-intro}

Logics for strategic reasoning, such as Alternating-time Temporal
Logic (\ATL)~\cite{DBLP:journals/jacm/AlurHK02} and Strategy Logic
(\SL)~\cite{chatterjee2010strategy,DBLP:journals/tocl/MogaveroMPV14},
are powerful languages to specify
complex synthesis problems for distributed systems and verify
strategic abilities in multi-agent systems.  Strategy Logic in
particular is very expressive: it can express the
existence of distributed strategies  satisfying important game-theoretic
solution concepts such as Nash equilibria or subgame-perfect
equilibria; and since model-checking algorithms for \SL can usually
provide witnesses of distributed strategies when they exist, such algorithms 
constitute generic solutions for a range of synthesis problems such as distributed
synthesis~\cite{PR90,kupermann2001synthesizing,DBLP:conf/lics/FinkbeinerS05}
or rational
synthesis~\cite{fisman2010rational,DBLP:conf/icalp/ConduracheFGR16,DBLP:journals/amai/KupfermanPV16,filiot2018rational}.

Most works on such logics have focused on finite-state systems, but in
the recent years a line of work has considered the model-checking
problem for \ATL and \SL on a class of infinite systems that plays an
important role in program verification, namely those generated by
pushdown systems. Pushdown systems are finite-state transition systems
equip\-ped with a stack. Because these systems can capture the flow of
procedure  calls and returns in programs~\cite{jones1977even}, many
problems in formal methods that were initially concerned with
finite-state systems have been studied and solved on such infinite
systems: model checking temporal
logics~\cite{bouajjani1997reachability,finkel1997direct,DBLP:journals/iandc/EsparzaKS03},
solving reachability and parity
games~\cite{walukiewicz2001pushdown,serre2003note,cachat2003games,serrePHD,piterman2004global,hague2009winning},
module checking~\cite{bozzelli2010pushdown,aminof2013pushdown}, and
more recently model checking of logics for strategic
reasoning~\cite{murano2015pushdown,chen2016global,DBLP:conf/ijcai/ChenSW16,chen2017model}.

Two of these works~\cite{aminof2013pushdown,chen2017model} consider
pushdown systems with \emph{imperfect information}, \ie
systems where  players or components  may not
observe perfectly the state of the system. Imperfect information plays an
important role in game theory and distributed systems, but it
usually increases greatly the complexity of analysing such systems:
 already for finite-state systems, multiplayer reachability
games are undecidable when no assumption is made on the relative
information of the players~\cite{DBLP:conf/focs/PetersonR79}.  To
retrieve decidability, a common restriction is to consider systems
with \emph{hierarchical information}, \ie where the players can be
totally ordered according to how well they observe the system.  This restriction has been used to establish results on
multiplayer
games~\cite{peterson2002decision,DBLP:journals/acta/BerwangerMB18} and
distributed
synthesis~\cite{PR90,kupermann2001synthesizing,DBLP:conf/lics/FinkbeinerS05},
and more recently on the model-checking problem for \SLi, an extension of Strategy Logic to
the imperfect-information setting~\cite{DBLP:conf/lics/BerthonMMRV17}.
This result states that the model-checking problem for \SLi is
decidable as long as strategies quantified deeper in the formula
observe the system better than those higher up in the syntactic
tree. {We show that  this result can be extended to
infinite arenas generated by collapsible pushdown systems, as long as the stack is
visible to all players, who thus have imperfect information only on
the control states. This higher-order extension of pushdown system permits to capture higher-order procedure calls (see \eg \cite{HagueMOS17,BroadbentCHS12,BroadbentCHS13}), a feature embraced by many modern day programming languages such as \cpp, Haskell, OCaML, Javascript, Python, or Scala.}

We first consider the simpler case of pushdown systems.
We extend the approach followed
in~\cite{DBLP:conf/lics/BerthonMMRV17}, which consists in reducing the
model-checking problem for \SLi to that of \QCTLsi, an intermediary,
low-level logic introduced in~\cite{DBLP:conf/lics/BerthonMMRV17} as
an imperfect-information extension of
\QCTLs~\cite{DBLP:journals/corr/LaroussinieM14}, which itself extends
\CTLs with second-order quantification on atomic
propositions. In~\cite{DBLP:conf/lics/BerthonMMRV17}, \QCTLsi is
evaluated on finite \emph{compound Kripke structures}, which are
Kripke structures whose states are tuples of \emph{local states}, and
the second-order quantifiers are parameterised by an indication of
which components of states they can observe.
We introduce \emph{pushdown compound Kripke structures}, which
are compound Kripke structures equipped with a
 stack, 
and we show that the model-checking problem for \QCTLsi on such structures
is decidable for the \emph{hierarchical fragment} of \QCTLsi, where
innermost quantifiers observe better than outermost ones.

To prove this,  we generalise the automata construction
  from~\cite{DBLP:conf/lics/BerthonMMRV17}. Instead of alternating
  tree automata we naturally use alternating  pushdown tree automata
  (\APTA), introduced in~\cite{kupferman2002pushdown}.  The emptiness
  problem for these automata being
  undecidable~\cite{aminof2013pushdown}, we actually resort to the
  subclass of \emph{semi-alternating pushdown tree automata} (\SPTA). These
  automata were introduced in~\cite{aminof2013pushdown} to solve the
  module-checking problem on pushdown systems with imperfect
  information.  The idea is the following:
  in the automata constructions  considered, the stack of an
  automaton is always used to simulate that of the (unfolding
  of the) pushdown system that it reads as input. In addition, the operations on the
  system's stack are  coded in the directions of the input tree.
  It follows that the content of an automaton's stack is determined by
  the node it visits, and thus all copies of an automaton that visit a same node in the input
  tree share the same stack content, unlike general
  alternating pushdown tree automata.  \SPTAs were introduced to exploit this property
  and obtain a simulation procedure (elimination of alternation), and
  thus a decidable class of \APTAs.
However, these automata are not closed
under an operation that is central  in the approach
from~\cite{DBLP:conf/lics/BerthonMMRV17} that we
generalise: the   \emph{narrowing} operation.

Narrowing is an operation on tree
automata that was introduced by Kupferman and Vardi to deal with imperfect
information in the automata approach to \LTL
synthesis~\cite{kupferman1999church,kupermann2001synthesizing}.
Intuitively, if a tree automaton works on $X\times Y$-trees (\ie
trees where nodes are words over $X\times Y$), its
narrowing to $X$ is an automaton that works on $X$-trees and can thus
guess a strategy that observes only $X$.  We generalise this operation
to pushdown tree automata. This presents no difficulty, but it turns
out that \SPTAs are not closed under narrowing: if an \SPTA sends two
copies of itself in two directions $(x,y)$ and $(x,y')$ with different
operations on the stack (which is possible in an \SPTA if $y\neq y'$), in its narrowing to $X$ these two copies take
the same direction $x$ and thus arrive in the same node with two
different stack contents. To solve this problem, we identify a
subclass of semi-alternating pushdown tree automata that is stable
under narrowing,
and we prove that it is also stable under simulation and
projection (the latter is trivial), the two other main operations 
involved in the automata construction.

The idea is the following:  in \SPTAs the operation on the stack
can depend on the direction taken in the input tree. 
We observe
 that actually, since the automata we build work on  unfoldings of  pushdown systems whose stack
 operations are coded as part of the directions, these stack operations are
 determined by a specific component of the directions. 
And moreover, because the
stack is visible, this component coding stack operations is never erased by
the narrowing operations we perform. We say that an \APTA
working on $X\times Y$-trees is \emph{$X$-guided} if stack operations
are determined by the $X$ component of the direction taken,
and we will use the fact that
 if an automaton working on $X\times Y\times Z$-trees is $X$-guided, then its
 narrowing to $X\times Y$ is also $X$-guided.

 {In the higher-order case, we follow the same road map. The
   main technical difficulty arises when defining regular
   labelling functions, which are tools to describe the atomic
   propositions satisfied in a given configuration of the collapsible
   pushdown system. For that we follow the approach from
   \cite{BCOS10} {which, in particular,} permits to rely on a closure
   property of the model of alternating collapsible pushdown automata
   to solve most of the technical difficulties.}

 \subsection*{Related work}
Pushdown systems with imperfect information and visible stack were
 considered in~\cite{aminof2013pushdown}, where it is proved that
 module checking
is undecidable if the stack is not visible. This is also the
case of the model-checking problem for \SLi, as it subsumes module
checking. The only existing work on logics for strategic reasoning
on pushdown systems with imperfect information
is~\cite{chen2017model}. The logics it considers  are incomparable to
\SLi: they involve epistemic operators, but are based on \ATL instead
of the richer \SL; also, while we work in the setting of perfect recall, they consider memoryless players, which makes it
possible to make less restrictive assumptions on the visibility of the
stack while retaining decidability.

\subsection*{Plan}
 We start in Section~\ref{sec-QCTLi} 
 by defining \QCTLsi and pushdown compound Kripke structures.
{Section~\ref{sec-direction-guided} contains the main conceptual
 novelty of this work, which is the introduction of direction-guided
 pushdown automata, and the proof that
they are stable under projection, simulation and narrowing.
In Section~\ref{sec-proof}  we use these automata to extend the automata
construction from~\cite{DBLP:conf/lics/BerthonMMRV17} to the case of
pushdown systems, and obtain
 our decidability result
 for  \QCTLsi model checking (Theorem~\ref{theo-decidable-QCTLi}).}
 We then apply this result to Strategy Logic with imperfect
 information. In Section~\ref{sec-sl} we
recall its syntax, define its semantics on pushdown game arenas,
 and we show how the
 hierarchy-preserving reduction from  \SLi to \QCTLsi can be extended
 to the pushdown setting, which 
 entails our main result {on pushdown arenas}
(Theorem~\ref{theo-SLi}).  {Finally, in
Section~\ref{sec-ho} we {show how to generalise this result} to a much more general
case in which pushdown {arenas} are replaced with
collapsible pushdown {arenas} while preserving
decidability (Theorem~\ref{theo-SLi-HO}).}

\section{\QCTLs with imperfect information}
\label{sec-QCTLi}

We start by recalling the syntax and semantics of \QCTLsi. The
definitions are as
in~\cite{DBLP:conf/lics/BerthonMMRV17},
except that the models are now \emph{pushdown} compound Kripke structures
instead of finite ones.

\subsection*{Preliminaries} As usual we write $A^*$ (\resp $A^+$, $A^\omega$) for the set of
finite (\resp finite nonempty, infinite) words over some finite alphabet $A$.
The \emph{length} of a finite word $w=w_{0}w_{1}\ldots
w_{n}$ is $|w|\egdef n+1$, and
 $\last(w)$ is the last
  letter .
Given a finite (resp. infinite) word $w$ and $0 \leq i < |w|$ (resp. $i\in\setn$), we let $w_{i}$ be the
letter at position $i$ in $w$, $w_{\leq i}$ is the prefix of $w$ that
ends at position $i$ and $w_{\geq i}$ is the suffix of $w$ that starts
at position $i$. 
The domain
of a mapping $f$ is written $\dom(f)$, for a relation $\relation\subseteq A\times B$ and $a\in A$,
$\relation(a)\egdef\{b\in B \mid (a,b)\in \relation\}$,   and for $n\in\setn$ we let
$[n]\egdef\{i \in \setn: 1 \leq i \leq n\}$.

\subsection{\QCTLsi Syntax}
\label{sec-syntax-QCTLi}

For the rest of the paper we fix
a finite  set of
\defin{atomic propositions} $\APf$,
and some natural number $n\in\setn$ which
parameterises the logic \QCTLsi, and which is the number of components
in states of the models.
We also  let $\{\setlstates_{i}\}_{i\in [n]}$ be a family of $n$ disjoint  sets of
\defin{local states}.
In \QCTLsi each quantifier on atomic propositions is parameterised by
a set of indices that represents which components of each state the
quantifier observes; it thus defines the
``observation'' of that quantifier.
Accordingly, a set $\cobs \subseteq [n]$ is called a
\defin{concrete observation} (to distinguish it from observation
symbols $\obs$ used in \SLi, see Section~\ref{sec-sl}).

\begin{defi}
  \label{def-syntax-QCTLsi}
  The syntax of \QCTLsi is defined by the following grammar:
  \begin{align*}
  \phi\egdef &\; p \mid \neg \phi \mid \phi\ou \phi \mid \E \psi \mid
  \existsp[p]{\cobs} \phi\\
    \psi\egdef &\; \phi \mid \neg \psi \mid \psi\ou \psi \mid \X \psi \mid
  \psi \until \psi
\end{align*}
where $p\in\APf$ 
and $\cobs\subseteq [n]$. 
\end{defi}

Formulas of type $\phi$ are \emph{state formulas}, those of type $\psi$
are \emph{path formulas}, and \QCTLsi consists of all the state formulas
defined by the grammar.

The set of \emph{quantified
  propositions} $\APq(\phi)\subseteq\APf$ of a \QCTLsi formula $\phi$ is the set of atomic propositions $p$ such that
$\phi$ has a subformula of the form $\existsp[p]{\cobs}\phi'$. We also
define the set of \emph{free propositions} $\APfree(\phi)\subseteq\APf$ as the set
of atomic propositions that have an occurrence which is not under the scope of any quantifier
of the form $\existsp[p]{\cobs}$ 
Without loss of generality we will assume that
$\APq(\phi)\inter\APfree(\phi)$ is empty and that each $p\in\APq(\phi)$ is
quantified at most once in $\phi$. 

\subsection{Compound Kripke structures}
\label{sec-CKS}


Compound Kripke structures~\cite{DBLP:conf/lics/BerthonMMRV17} are  Kripke structures where 
states  are tuples $\sstate=(\lstate_1,\ldots,\lstate_n)$ in which the
$\lstate_i$ are \emph{local states}. A concrete observation
$\cobs\subseteq [n]$ indicates the indices of the local states
observed by a propositional
quantifier. Unlike~\cite{DBLP:conf/lics/BerthonMMRV17}, here we define
the semantics of \QCTLsi on potentially \emph{infinite} structures,
that will be generated first by finite-state pushdown compound Kripke
structures that we introduce in Section~\ref{sec-PCKS}, \bmchanged{and later 
  by (higher-order) collapsible pushdown compound Kripke structures in Section~\ref{sec-ho}.}
  
\begin{defi}
A \emph{compound Kripke structure}, or \CKS, {over local states $\{\setlstates_i\}_{i\in [n]}$}   is a tuple 
$\CKS=(\setstates,\relation,\lab,\sstate_\init)$ where
\begin{itemize}
\item $\setstates\subseteq \prod_{i\in [n]}\setlstates_i$  is a set of
\emph{states},  
\item $\relation\subseteq\setstates\times\setstates$ is a
left-total \emph{transition
relation}, 
\item $\lab:\setstates\to 2^{\APf}$ is a \emph{labelling function} and
\item $\sstate_\init \in \setstates$ is an \emph{initial state}.
\end{itemize}
\end{defi}

A \emph{path} in $\CKS$  is an infinite sequence of states
$\spath=\sstate_{0}\sstate_{1}\ldots$ such that
$\sstate_0=\sstate_\init$ and for all $i\in\setn$,
$(\sstate_{i},\sstate_{i+1})\in \relation$. 
A \emph{partial path} is a finite non-empty prefix of a path.

\subsection{\QCTLsi semantics}
\label{sec-QCTLsi-semantics}

\QCTLsi is interpreted on
infinite trees, which represent unfoldings of \CKSs.
Let $\Dirtree$ be a (possibly infinite) 
set of \emph{directions}. 
An \defin{$\Dirtree$-tree} $\tree$ 
 is a
set of words $\tree\subseteq \Dirtree^+$ such that
$\bm{(1)}$  there exists $\racine\in\Dirtree$,  called the
    \emph{root} of $\tree$, such that each
    $\noeud\in\tree$ starts with $\racine$; 
$\bm{(2)}$ if $\noeud\cdot\dir\in\tree$ and $\noeud\cdot\dir\neq\racine$, then
    $\noeud\in\tree$; and
$\bm{(3)}$ if $\noeud\in\tree$ then there exists $\dir\in\Dirtree$ such that $\noeud\cdot\dir\in\tree$.

The elements of a tree $\tree$ are called \emph{nodes}.  
A \emph{\tpath} in $\tree$ is an infinite sequence of nodes $\tpath=\noeud_0\noeud_1\ldots$
such that for all $i\in\setn$, $\noeud_{i+1}=\noeud_i\cdot\dir$ for
some $\dir\in\Dirtree$, and $\tPaths(\noeud)$ is the set of \tpaths
 that start in node $\noeud$.
 An $\Dirtree$-tree $\tree$ is \emph{complete} if for every $\noeud \in
 \tree$ and  $\dir \in \Dirtree$,  $\noeud \cdot \dir \in \tree$. 
An \defin{$\APf$-labelled $\Dirtree$-tree}, or
\defin{$(\APf,\Dirtree)$-tree} for short, is a pair
$\ltree=(\tree,\lab)$, where $\tree$ is an $\Dirtree$-tree called the
\emph{domain} of $\ltree$ and
$\lab:\tree \rightarrow 2^{\APf}$ is a \emph{labelling}.
A \emph{pointed labelled tree} is a pair $(\ltree,\noeud)$ where
 $\noeud$ is a node of $\ltree$.

Let $p\in\APf$ and $\tree$ a tree. A \emph{$p$-\labeling} for $\tree$ is a mapping
$\plab:\tree\to \{0,1\}$ that indicates in which nodes $p$ holds, and
for a \labeled tree $\ltree=(\tree,\lab)$, the $p$-\labeling of $\ltree$ is
the $p$-\labeling $\noeud \mapsto 1$ if $p\in\lab(\noeud)$, 0 otherwise. 
The composition of a \labeled tree $\ltree=(\tree,\lab)$ with a
$p$-\labeling $\plab$ for $\tree$ is defined as
$\ltree\prodlab\plab\egdef(\tree,\lab')$, where
$\lab'(\noeud)=\lab(\noeud)\union \{p\}$ if $\plab(\noeud)=1$, and
$\lab(\noeud)\setminus \{p\}$ otherwise.
A $p$-\labeling for a labelled tree $\ltree=(\tree,\lab)$ is a
$p$-\labeling for its domain $\tree$.


Let $\Dirtree$ and $\Dirtreea$ be two sets, and let $(\dir,\dira)\in\Dirtree\times\Dirtreea$.
 The
 \defin{$\Dirtree$-narrowing} of $(\dir,\dira)$ is
${\projI[\Dirtree]{(\dir,\dira)}}\egdef \dir$.
This definition extends naturally to words and trees over
$\Dirtree\times\Dirtreea$.
For $I\subseteq [n]$, we let $\Dirtreei\egdef\prod_{i\in
I}\setlstates_{i}$ if $I\neq\emptyset$ and
$\Dirtreei[\emptyset]\egdef\{\blank\}$, where $\blank$ is a special symbol.
 For ${I,J\subseteq [n]}$ and $\dirz=(\lstate_{i})_{i\in
   I}\in\Dirtreei[I]$,
 we also let 
 \[{\projI[{J}]{\dirz}}\egdef
\projI[{\Dirtreei[ {I\cap J}]}]{\dirz}{\in \Dirtreei[I\cap J],}\]
 where $\dirz$ is seen as a pair $\dirz=(\dir,\dira)\in
 \Dirtreei[{I\cap J}]\times \Dirtreei[{I\setminus J}]$, \ie we apply the
 above definition with
 $\Dirtree=\Dirtreei[{I\cap J}]$ and $\Dirtreea=\Dirtreei[{I\setminus
   J}]$\footnote{Since sets $\setlstates_{i}$ are disjoint, the
 ordering of local states in $\dirz$ is indifferent and thus this is well defined.}. 
 We extend this definition to words and trees. 

 To define the semantics of quantifier $\exists^{\cobs}p$
 we need to
 define what it means for a $p$-labelling of a tree to be
 $\cobs$-uniform.
For $\cobs \subseteq [n]$ and $I \subseteq [n]$,
two tuples $\dir,\dir'\in\Dirtreei[I]$ are \defin{$\cobs$-indistinguishable},
written $\dir\oequiv\dir'$, if 
$\projI[{I\cap\,\cobs}]{\dir}=\projI[{I\cap\,\cobs}]{\dir'}$.
 Two words
   $\noeud=\noeud_{0}\ldots\noeud_{i}$ and
   $\noeud'=\noeud'_{0}\ldots\noeud'_{j}$ over alphabet $\Dirtreei[I]$ are
   \defin{$\cobs$-indistinguishable}, written $\noeud\oequivt\noeud'$, if
   $i=j$ and for all $k\in \{0,\ldots,i\}$ we have
   $\noeud_{k}\oequiv\noeud'_{k}$.
   Finally,
a $p$-\labeling $\plab$ for {an $\Dirtreei$-tree} $\tree$ is \defin{$\cobs$-uniform} if for all
    $\noeud,\noeud'\in\tree$, $\noeud\oequivt\noeud'$ implies
 $\plab(\noeud)=\plab(\noeud')$. 

\begin{defi}
We define by induction the satisfaction relation $\modelst$ of
\QCTLsi. Let $I\subseteq [n]$, let $\ltree=(\tree,\lab)$ be
an $\APf$-labelled $\Dirtreei$-tree,
$\noeud$  a node and $\tpath$  a path in $\tree$:
\begingroup
  \addtolength{\jot}{-2pt}
\begin{alignat*}{3}
  \ltree,\noeud\modelst & 	\,p 			&& \mbox{ if } &&\quad p\in\lab(\noeud)\\
  \ltree,\noeud\modelst & 	\,\neg \phi		&& \mbox{ if } && \quad\ltree,\noeud\not\modelst \phi\\
  \ltree,\noeud\modelst & 	\,\phi \ou \phi'		&& \mbox{ if } &&\quad \ltree,\noeud \modelst \phi \mbox{ or    }\ltree,\noeud\modelst \phi' \\
  \ltree,\noeud\modelst & 	\,\E\psi			&& \mbox{ if } &&\quad \exists\,\tpath\in\tPaths(\noeud) \mbox{      s.t. }\ltree,\tpath\modelst \psi \\
  \ltree,\noeud\modelst & \,\existsp{\cobs} \phi && \mbox{ if }
  && \quad \exists\,\plab \mbox{ a $\cobs$-uniform $p$-\labeling for
    $\ltree$} \mbox{ such that 
  }\ltree\prodlab\plab,\noeud\modelst\phi\\
\ltree,\tpath\modelst &		\,\phi 			&& \mbox{ if } &&\quad \ltree,\tpath_{0}\modelst\phi \\ 
\ltree,\tpath\modelst &		\,\neg \psi 		&& \mbox{ if }
&& \quad \ltree,\tpath\not\modelst \psi \\ 
\ltree,\tpath\modelst & \,\psi \ou \psi'	\quad		&& \mbox{ if } && \quad\ltree,\tpath \modelst \psi \mbox{ or }\ltree,\tpath\modelst \psi' \\ 
\ltree,\tpath\modelst & \,\X\psi 				&& \mbox{ if } && \quad\ltree,\tpath_{\geq 1}\modelst \psi \\ 
\ltree,\tpath\modelst & \,\psi\until\psi' 		&& \mbox{ if }
&& \quad\exists\, i\geq 0 \mbox{ s.t.    }\ltree,\tpath_{\geq
  i}\modelst\psi' \text{ and } \forall j \text{ s.t. }0\leq j <i,\; \ltree,\tpath_{\geq j}\modelst\psi
\end{alignat*}
\endgroup
\end{defi}

  Let $\CKS=(\setstates,\relation,\lab,\sstate_\init)$ be a compound
  Kripke structure over $\APf$.  The
  \defin{tree-unfolding of $\CKS$} is the
  $(\APf,\setstates)$-tree $\unfold[\CKS]{\sstate}\egdef (\tree,\lab')$,
  where $\tree$ is the set of all partial paths in $\CKS$, and for every $\noeud\in\tree$,
  $\lab'(\noeud)\egdef \lab(\last(\noeud))$.
 We write 
$\CKS \modelst \phi$ if
$\unfold[\CKS]{\sstate},\sstate_\init \models \phi$.

  \subsection{Pushdown compound Kripke structures}
  \label{sec-PCKS}

  We now focus on infinite compound Kripke structures
generated by \emph{pushdown} compound Kripke structures, which are compound Kripke structures equipped with
a (visible) stack.

\begin{defi}
  \label{def-pcks}
A \emph{pushdown compound Kripke structure}, or \PCKS, {over
local states $\{\setlstates_i\}_{i\in [n]}$} is a tuple 
$\PCKS=(\stacka,\setstates,\relation,\lab,\sstate_\init)$ where
\begin{itemize}
    \item $\stacka$ is a finite stack alphabet  together with a  bottom
       symbol $\stackb\notin\stacka$, and we let $\stacka_\stackb=\stacka\cup\{\stackb\}$;
\item $\setstates\subseteq \prod_{i\in [n]}\setlstates_i$  is a finite
  set of
states; 
\item $\relation\subseteq\setstates\times\stacka_\stackb\times\setstates\times{\stacka_\stackb}^*$ is a
 transition
relation; 
\item $\lab:\setstates\times\stacka^*\cdot\stackb\to 2^{\APf}$ is a
  regular labelling function (defined below);
\item $\sstate_\init \in \setstates$ is an initial state.
\end{itemize}
\end{defi}

We require that the bottom symbol can never be removed
nor pushed: for any $\sstate\in\setstates$
 one has
$\relation(\sstate,\stackb)\subseteq
\setstates\times{\stacka}^*\cdot\stackb$ (the bottom symbol is never
removed), and for every $\stacks\in\stacka$,
$\relation(\sstate,\stacks)\subseteq \setstates\times{\stacka}^*$ (the
bottom symbol is never pushed).

A \defin{regular labelling function} is given as a set of finite word automata
$\wautop{\sstate}$ over alphabet $\stacka$, one for each $p\in\APf$
and each $\sstate\in\setstates$. They define the labelling function
that maps to each state $\sstate\in\setstates$ and stack content
$\stackc\in\stacka^*\cdot\stackb$ the set $\lab(\sstate,\stackc)$ of
all atoms $p$ such that $\stackc$ belongs to
$\lang(\wautop{\sstate})$, the language accepted by
$\wautop{\sstate}$.

\begin{rem}\label{rk-regular-labelling}
  Because of the definition of regular labelling function, whether an
  atomic proposition holds in a configuration depends not only on the
  control state but on the whole content of the stack. We believe that
  it is important to \bmchanged{be able to} express properties about the whole
  stack content, as the latter reflects the recursive calls of a
  system.
	
	The choice of restricting to \emph{regular} properties is for
        decidability issues. However it is already expressive, as for
        instance, it permits to capture all sets of configurations
        that one can define in popular logics such as the monadic
        second order logic or the modal $\mu$-calculus. {Such regular
        labelling functions were used  for
        instance in~\cite{DBLP:journals/iandc/EsparzaKS03}.}
      \end{rem}
      

  
A \defin{configuration} is a pair
$\config=\conf{\sstate}{\stackc}\in \setstates\times({\stacka}^*\cdot\stackb)$
where $\sstate$ is the current state and $\stackc$ the current
content of the stack. From  
configuration $\conf{\sstate}{\stacks\cdot\stackc}$ the system can move to
a configuration $\conf{\sstate'}{\stackc'\cdot\stackc}$ if
$(\sstate,\stacks,\sstate',\stackc')\in\relation$, which
we write
$\conf{\sstate}{\stacks\cdot\stackc}\conftrans{}\conf{\sstate'}{\stackc'\cdot\stackc}$. We assume that
for every configuration $\conf{\sstate}{\stackc}$ there exists at
least one configuration $\conf{\sstate'}{\stackc'}$ such that
$\conf{\sstate}{\stackc}\conftrans{}\conf{\sstate'}{\stackc'}$.
A \defin{path} in $\PCKS$  is an infinite sequence of configurations
$\spath=\config_{0}\config_{1}\ldots$ such that
$\config_0=\conf{\sstate_\init}{\stackb}$ and
 for all $i\in\setn$,
$\config_i\conftrans{}\config_{i+1}$. 
A \defin{partial path} is a finite non-empty prefix of a path.
We let $\Paths(\PCKS)$
(resp. $\PPaths(\PCKS)$) be the set of all
paths (resp. partial paths) in $\PCKS$.

\begin{defi}
  \label{def-generated-CKS}
A \PCKS $\PCKS=(\stacka,\setstates,\relation,\lab,\sstate_\init)$
over $\{\setlstates_i\}_{i\in [n]}$  generates an infinite \CKS
$\CKS_\PCKS=(\setstates',\relation',\lab',\sstate'_\init)$ over 
$\{\setlstates_i\}_{i\in [n+1]}$, where
\begin{itemize}
\item $\setlstates_{n+1}={\stacka}^*\cdot\stackb$,
\item $\setstates'=\setstates\times{\stacka}^*\cdot\stackb$, 
\item $(\sstate',\stackc')\in\relation'(\sstate,\stackc)$ if 
$(\sstate,\stackc)\conftrans{}(\sstate',\stackc')$,
\item  $\lab'=\lab$ and
\item $\sstate'_\init=(\sstate_\init,\stackb)$.
\end{itemize}
\end{defi}

We write $\PCKS\models\phi$ if $\CKS_\PCKS\models\phiprime$, where $\phiprime$ is
obtained from $\phi$ by replacing each concrete observation $\cobs\subseteq
[n]$ with $\cobs'=\cobs\union\{n+1\}$. \bmchanged{This reflects the fact that the
stack content is  visible to all  quantifiers  in $\phi$.}

There is a reduction from the model-checking problem for MSO with
equal-level predicate on the infinite binary tree to the
model-checking problem for \QCTLsi, so that this problem is
undecidable already on \emph{finite} compound Kripke structures~\cite{DBLP:conf/lics/BerthonMMRV17}.
However it is proved in~\cite{DBLP:conf/lics/BerthonMMRV17} that the
problem is decidable for the fragment of \emph{hierarchical
  formulas}. We now recall this notion, and then we generalise this
result to the case of pushdown compound Kripke structures.

\begin{defi}
  \label{def-hierarchical}
  A \QCTLsi formula $\phi$ is \defin{hierarchical} if for every
  subformula 
  $\phi_{1}=\existsp[p_{1}]{\cobs_{1}}\phi'_{1}$ of $\phi$ and
subformula  $\phi_{2}=\existsp[p_{2}]{\cobs_{2}}\phi'_{2}$ of $\phi'_1$,  
 we have $\cobs_{1}\subseteq\cobs_{2}$.
\end{defi}

A formula is thus hierarchical if innermost propositional
quantifiers observe at least as much as  outermost ones.
We let \QCTLsih be the set of hierarchical \QCTLsi formulas.

\begin{thm}
  \label{theo-decidable-QCTLi}
Model checking \QCTLsih on pushdown compound Kripke structures is decidable.
\end{thm}

Before proving this result in Section~\ref{sec-proof}, we introduce a new subclass of alternating pushdown tree automata, and
show that it is stable under the operations that are required for the
automata construction on which the proof relies. 

\section{A subclass of pushdown tree automata}
\label{sec-direction-guided}


{In this section we present the class of direction-guided pushdown tree
automata, a subclass of alternating pushdown tree automata that is
decidable and stable under the operations needed to generalise the construction
  from~\cite{DBLP:conf/lics/BerthonMMRV17} to the case of pushdown
  systems.}


\subsection{Alternating pushdown tree automata}
\label{sec-tree-automata}

  We  recall alternating pushdown parity tree automata~\cite{ladner1984alternating,kupferman2002pushdown}.
  Because it is sufficient for our needs and simplifies definitions,
  we assume that all input trees are complete trees.
  
  For a set $Z$, $\boolp(Z)$ is the set of formulas built from the
  elements of $Z$ as atomic propositions using the connectives $\ou$
  and $\et$, and with $\top,\perp\, \in \boolp(Z)$.  For $\APf$ a
  finite set of atomic propositions and $\Dirtree$ a finite set of
  directions, an \defin{alternating pushdown tree automaton
    (\APTA) on $(\APf,\Dirtree)$-trees} is a tuple
  $\auto=(\stacka,\tQ,\tdelta,\tq_{\init},\couleur)$ where $\stacka$
  is a finite stack alphabet with a special bottom symbol $\stackb$,
  $Q$ is a finite set of states, $\tq_{\init}\in \tQ$ is an initial
  state,
  $\tdelta : \tQ\times 2^{\APf}\times \stacka \rightarrow
  \boolp(\Dirtree\times \tQ\times\stacka^*)$ is a transition function (that never pushes nor removes the bottom symbol $\stackb$),
  and $\couleur:\tQ\to \setn$ is a colouring function.  Atoms in
  $\boolp(\Dirtree\times\tQ\times\stacka^*)$ are written between
  brackets, such as $[x,\tq,\stackc]$.
A \defin{nondeterministic pushdown tree
    automaton (\NPTA)} is an alternating pushdown tree automaton
  $\nauto=(\stacka,\tQ,\tdelta,\tq_{\init},\couleur)$ such that for
  every $\tq\in \tQ$, $a\in 2^{\APf}$ and  $\stacks\in\stacka$,
  $\tdelta(\tq,a,\stacks)$ is written in disjunctive normal form and
  for every direction $\dir\in \Dirtree$, each disjunct contains
  exactly one element of $\{\dir\}\times Q\times \stacka^*$.

  We define acceptance of a tree by an \APTA in a given initial node
  and a given initial stack content via a two-player (Eve and Adam)
  turn-based perfect-information parity game. Due to space
  constraints, we do not give a formal definition of parity games but
  we refer the reader to,
  e.g.,~\cite{DBLP:journals/tcs/Zielonka98,LNCS2500} for definitions
  and classical concepts such as strategies and winning positions.  Let
  $\SPTA=(\stacka,\tQ,\tdelta,\tq_\init,\couleur)$ be an \APTA over
  $(\APf,\Dirtree)$-trees, let $\ltree=(\tree,\lab)$ be such a tree,
  let $\noeud_\init \in \tree$ be a starting node and let
  $\stackc_\init\in\stacka^*$ be an initial stack content.  We define
  the parity game $\tgame{\SPTA}{\ltree}{\noeud_\init,\stackc_\init}$
  whose set of positions is
  $\tree\times \tQ\times \stacka^* \times \boolp
  (\Dirtree\times \tQ \times \stacka^*)$, and the initial position is
  $\pos_{\init}=(\noeud_\init,\tq_\init,\stackc_\init,\tdelta(\tq_\init,\stacks,\lab(\noeud_\init)))$,
  where $\stacks\in\stacka$ is the top symbol of $\stackc_\init$.  A
  position $(\noeud,\tq,\stackc,\pform)$ belongs to Eve if $\pform$ is
  of the form $\pform_1\vee \pform_2$, otherwise it belongs to Adam
  (note that if $\pform$ is of the form $[\dir,\tq',\stackc']$ then
  there is no choice to be made).  Moves in
  $\tgame{\tauto}{\ltree}{\noeud_\init,\stackc_\init}$ are defined by
  the following rules:
    \begin{align*}
 (\noeud,\tq,\stackc,\pform_1 \;\op\; \pform_2) & \move (\noeud,\tq,\stackc,\pform_i) 
 &&\mbox{where } 
 \op \in\{\vee,\wedge\} \mbox{ and } i\in\{1,2\},  \\[3pt]
     (\noeud,\tq,\stacks\cdot\stackc,[\dir,\tq',\stackc']) & \move (\noeud\cdot
            \dir,\tq',\stackc'\cdot\stackc,\tdelta(\tq',\lab(\noeud\cdot\dir),\stacks')) 
 &&\mbox{where }\stacks' \mbox{ is the top  of }\stackc'\cdot\stackc
    \end{align*}

\ie Eve resolves existential/disjunctive choices in the formula while
Adam resolves  universal/conjunctive choices. 

Positions of the form $(\noeud,\tq,\stackc,\top)$ and
$(\noeud,\tq,\stackc,\perp)$ are deadlocks, winning for Eve and Adam
respectively.  Finally, the colouring function (used to define the
parity condition) is
$\couleur'(\noeud,\tq,\stackc,\pform)=\couleur(\tq)$.

A pointed tree $(\ltree,\noeud)$
is \defin{accepted} by $\tauto$ with initial stack content
$\stackc$ if Eve has a winning strategy in
$\tgame{\tauto}{\ltree}{\noeud,\stackc}$, \ie she has a way of
playing such that whatever the choices of Adam are, the resulting play
either ends up in a winning deadlock for her or is such that the
largest colour  visited infinitely often is even.
We also let   $\lang(\tauto,\stackc)$ be the set of pointed trees accepted by
$\tauto$ with initial stack content $\stackc$.


{Finally, classic \ATAs and \NTAs (without pushdown store)
  are obtained by removing from the above definitions all components
  referring to the pushdown store. For instance, an \ATA is a tuple
  $\auto=(\tQ,\tdelta,\tq_{\init},\couleur)$ with a transition function
 of type $\tdelta : \tQ\times 2^{\APf} \rightarrow
  \boolp(\Dirtree\times \tQ)$.}
   
\subsection{Direction-guided pushdown tree automata}
\label{subsec-direction-guided}


{We recall how semi-alternating pushdown tree automata are
defined by constraining
the  behaviour of the stack in \APTAs~\cite{aminof2013pushdown}, and then we constrain it further by
letting stack operations depend only on precise components of directions taken in the
input tree. We call the resulting class of automata \emph{direction-guided} pushdown tree automata.}

  A \defin{semi-alternating pushdown tree automaton (\SPTA)} is an
  \APTA such
  that for all $\tq_1,\tq_2\in\tQ$, $\stacks\in\stacka$ and
  $a\in 2^\APf$, if $[x,\tq_1',\stackc_1]$ appears in
  $\tdelta(\tq_1,a,\stacks)$ and $[x,\tq_2',\stackc_2]$ appears in
  $\tdelta(\tq_2,a,\stacks)$, then $\stackc_1=\stackc_2$: whenever two
  copies of the automaton, possibly in different states, read the same
  input with the same symbol on the top of the stack, and move in the
  same direction,  they must push the same thing on the stack.
{The transition function of an \SPTA can be split
into a \emph{state transition function} $\tdeltastate: \tQ \times 2^{\APf}\times \stacka \rightarrow
  \boolp(\Dirtree\times \tQ)$ and a \emph{stack update
    function} $\tdeltastack: 2^{\APf}\times\stacka\times\Dirtree\to \stacka^*$ such that for all
  $(\tq,a,\stacks)\in\tQ\times 2^{\APf}\times \stacka$ we have
  $\tdelta(\tq,a,\stacks)=\tdeltastate(\tq,a,\stacks)$, in which each
  $[\dir,\tq']$ is replaced with
  $[\dir,\tq',\tdeltastack(a,\stacks,x)]$
  (see~\cite{aminof2013pushdown} for details).}

{We now refine this definition to capture semi-alternating
  automata whose stack operations do not depend on the whole directions,
but only on precise components of the directions.}

\begin{defi}
  \label{def-guided}
{An \APTA $\auto=(\stacka,\tQ,\tdelta,\tq_{\init},\couleur)$ over $\Dirtree\times \Dirtreea$-trees has an
  \defin{$\Dirtree$-guided stack}, or simply is \defin{$\Dirtree$-guided}, if there exists a
  function $\tdeltastack: 2^{\APf}\times\stacka\times\Dirtree\to \stacka^*$ such that for all
  $(\tq,a,\stacks)\in\tQ\times 2^{\APf}\times \stacka$, all atoms
  appearing in
  $\tdelta(\tq,a,\stacks)$ are of the form $[(\dir,\dira),\tq',\tdeltastack(a,\stacks,x)]$.}
\end{defi}
Note that $\Dirtree$-guided
  \APTAs are semi-alternating.




We will need three operations on tree automata: projection, to guess
valuations of atomic propositions, simulation, because
projection is defined only for nondeterministic automata, and
narrowing, to deal with imperfect information by hiding components of
 directions.
{It
 was established in~\cite{aminof2013pushdown} that \SPTAs can be
 nondeterminised, and we show that the projection and narrowing
 operations on classic tree automata can be
easily extended to pushdown tree automata.
We then notice that the simulation procedure presented
in~\cite{aminof2013pushdown} preserves $\Dirtree$-guidedness, and so
does projection, as well as narrowing if the $\Dirtree$ component is not
erased by the operation.}

\subsection{Projection}
\label{sec-projection}

Projection is defined in~\cite{rabin1969decidability} for classic
nondeterministic tree automata. {The construction is simple: the
automaton projected on atom $p\in\APf$ guesses, in every node of its
input, a valuation for $p$ in this node, and proceeds accordingly.  This construction is correct
because nondeterministic automata only visit each node at most
once. The construction and the proof of
correctness are indifferent to the pushdown aspect, so that we have
the following result. 

\begin{prop}
  \label{theo-projection}
  Given an \NPTA $\NTA$  and  $p\in\APf$, one can build  an \NPTA
  $\proj{\NTA}$ 
  such that for every pointed tree $(\ltree,\noeud)$ and initial stack
  content $\stackc_\init\in\stacka^*\cdot\stackb$,
 \[   (\ltree{,\noeud})\in\, \lang(\proj{\NTA},\stackc_\init) \mbox{\bigiff}
    \exists \plab \mbox{ a\, $p$-\labeling for $\ltree$
  s.t. }(\ltree\prodlab\plab,\noeud)\in\lang(\NTA,\stackc_\init) \]    
\end{prop}

\begin{proof}
Let $\nauto=(\stacka,\tQ,\tdelta,\tq_{\init},\couleur)$  be a
nondeterministic pushdown tree automaton, and let
$\proj{\NTA}=(\stacka,\tQ,\tdelta',\tq_{\init},\couleur)$ where for
all $(\tq,a,\stacks)\in\tQ\times 2^{\APf\setminus \{p\}}\times
\stacka$,
\[\tdelta'(\tq,a,\stacks)=\tdelta(\tq,a\setminus \{p\},\stacks) \vee
  \tdelta(\tq,a\cup \{p\},\stacks).\]

To see that this construction is correct, fix a pointed tree
$(\ltree,\noeud)$, and assume first that there
exists a $p$-\labeling $\plab$ for $\ltree$
such that $(\ltree\prodlab\plab,\noeud)\in\lang(\NTA,\stackc_\init)$, \ie
Eve has a winning strategy $\strat$ in the acceptance game
$\tgame{\NTA}{\ltree\prodlab\plab}{\noeud,\stackc_\init}$. Observe that
$\tgame{\proj{\NTA}}{\ltree}{\noeud,\stackc_\init}$ is essentially the same game,
except that Eve has additional choices to make: everytime a new node
$\noeuda$ is reached, Eve has to choose between
$\tdelta(\tq,a\setminus \{p\},\stacks)$ and
$\tdelta(\tq,a\cup \{p\},\stacks)$. A winning strategy for Eve in
this game is obtained by
letting her choose $\tdelta(\tq,a\cup \{p\},\stacks)$ if $\plab(\noeud)=1$, $\tdelta(\tq,a\setminus \{p\},\stacks)$
otherwise, and all her remaining choices follow $\strat$. In other
words, Eve guesses the $p$-\labeling $\plab$ and otherwise behaves
as in $\tgame{\NTA}{\ltree\prodlab\plab}{\noeud,\stackc_\init}$.

Now assume that $\proj{\NTA}$ accepts $(\ltree=(\tree,\lab),\noeud)$,
and let $\strat$ be a winning strategy for Eve in
$\tgame{\proj{\NTA}}{\ltree}{\noeud,\stackc_\init}$.  Since $\NTA$ is
nondeterministic, by construction $\proj{\NTA}$ is also 
nondeterministic and thus each node of $\ltree$ is visited exactly
once in the outcomes of $\strat$. More precisely, for each node
$\noeuda\in\ltree$ that is below $\noeud$, there is a unique position
of the form
$(\noeuda,\tq,\stacks\cdot\stackc,\tdelta'(\tq,\lab(\noeuda),\stacks))$ that
can be reached while Eve follows strategy $\strat$. In addition, by
definition of $\proj{\NTA}$, we have that
\[\tdelta'(\tq,\lab(\noeuda),\stacks)=\tdelta(\tq,\lab(\noeuda)\setminus \{p\},\stacks)\vee
\tdelta(\tq,\lab(\noeuda)\cup \{p\},\stacks).\]  We can thus define the
$p$-\labeling 
\[\plab:\noeuda\mapsto
  \begin{cases}
    0 & \mbox{if }\strat \mbox{ chooses the first disjunct,}\\
    1 & \mbox{otherwise.}
  \end{cases}
\]
It is then not hard to see that $\strat$ induces a winning strategy
for Eve in $\tgame{\NTA}{\ltree\prodlab\plab}{\noeud,\stackc_\init}$.
\end{proof}

\subsection{Simulation}
\label{sec-simulation}

It is proved in~\cite{aminof2013pushdown} that, unlike alternating
pushdown tree automata, semi-alternating ones can be nondeterminised.

\begin{thm}[\cite{aminof2013pushdown}]
\label{theo-simulation}
Given an \SPTA $\SPTA$, one can build 
an \NPTA $\NPTA$ 
such that for every initial stack content $\stackc\in\stacka^*\cdot\stackb$, \[\lang(\NPTA,\stackc)=\lang(\SPTA,\stackc).\] 
\end{thm}

We observe that the construction in~\cite{aminof2013pushdown}
for the simulation of \SPTAs preserves $\Dirtree$-guidedness, and thus
we can refine the above result as follows:

\begin{prop}
  \label{prop-stable-simulation}
Given an $\Dirtree$-guided \SPTA $\SPTA$, one can build an
$\Dirtree$-guided \NPTA $\NPTA$ 
such that for every initial stack content $\stackc\in\stacka^*\cdot\stackb$, $\lang(\NPTA,\stackc)=\lang(\SPTA,\stackc)$.   
\end{prop}

\begin{proof}
Let $\SPTA=(\stacka,\tQ,\tdelta,\tq_{\init},\couleur)$ be an
 \SPTA over $\Dirtree$-trees, and
  let $\tdeltastack: 2^{\APf}\times\stacka\times\Dirtree\to \stacka^*$
  be its stack update function.  The construction
from~\cite{aminof2013pushdown} goes as follows. First, they observe
  that an \SPTA $\SPTA$ induces, for every input tree $\ltree$, a
  decorated version $\ltree'$ of $\ltree$ where the label of each node
  is enriched with the top symbol of the  automaton'stack when
  it visits that node. This is well-defined because the automaton is
  semi-alternating. One can then build a classic \ATA $\tilde \ATA$
  (without pushdown stack) such that $\SPTA$ accepts $\ltree$ if and
  only if $\tilde \ATA$ accepts $\ltree'$. With a classic simulation
  procedure, one then obtains an \NTA
  $\tilde \NTA=(\tQ,\tdelta',\tq'_{\init},\couleur')$ equivalent to
  $\tilde \ATA$. It remains to define the \NPTA
  $\NPTA=(\stacka,\tQ',\tdelta'',\tq'_{\init},\couleur')$ where
  $\tdelta''(\tq,a,\stacks)$ is obtained from
   $\tdelta'(\tq,(a,\stacks))$ by replacing every $[\dir,\tq']$ with
   $[\dir,\tq',\tdeltastack(a,\stacks,\dir)]$.

Now, if the initial automaton $\SPTA$ works on
$\Dirtree\times\Dirtreea$-trees and is $\Dirtree$-guided, by
definition its stack update function $\tdeltastack$ does not
depend on the $\Dirtreea$-components. By the above construction, it is
also the case of the final \NPTA:   $\tdelta''(\tq,a,\stacks)$ is now obtained from
   $\tdelta'(\tq,(a,\stacks))$ by replacing every $[(\dir,\dira),\tq']$ with
   $[(\dir,\dira),\tq',\tdeltastack(a,\stacks,\dir)]$.

Finally, one can see that all the above arguments generalise
easily to the case of automata starting in a given node $\noeud$ with
a given initial stack content $\stackc$.
\end{proof}

\subsection{Narrowing}
\label{sec-narrowing}

{For the last operation, we first recall the widening operation
on trees, defined in~\cite{kupferman1999church}: given two sets of
directions $\Dirtree$ and $\Dirtreea$, for every $\Dirtree$-tree
$\tree$ with root $\dir\in\Dirtree$ and every
$\dira\in\Dirtreea$ we define the \defin{$\Dirtreea$-widening
  of $\tree$ rooted in $(\dir,\dira)$} as the
$\Dirtree\times\Dirtreea$-tree
\[\liftI[\Dirtree\times\Dirtreea]{\dira}{\tree}\egdef\{\noeud
\in (\dir,\dira)\cdot (\Dirtree\times\Dirtreea)^*
\mid \projI[\Dirtree]{\noeud}\in\tree\}.\]

Also, for an
$(\APf,\Dirtree)$-tree $\ltree=(\tree,\lab)$ 
and an element $\dira\in\Dirtreea$, we let
\[\liftI[\Dirtree\times\Dirtreea]{\dira}{\ltree}\egdef
(\liftI[\Dirtree\times\Dirtreea]{\dira}{\tree},\lab'), \mbox{
  where }\lab'(\noeud)\egdef \lab(\projI[\Dirtree]{\noeud}).\]

We may write simply $\liftI[\Dirtree\times\Dirtreea]{}{\tree}$ and $\liftI[\Dirtree\times\Dirtreea]{}{\ltree}$ when the
choice of $\dira$ does not matter or is understood.
In particular, when referring to \emph{pointed} widenings of trees
  such as $(\liftI[\Dirtree\times\Dirtreea]{}{\ltree},\noeud)$, the
choice of the root is determined by $\noeud$: more precisely,
$\dira$ is taken to be the $\Dirtreea$-component
of the first direction in $\noeud$.}

 We now generalise the narrowing
 operation~\cite{DBLP:conf/lics/BerthonMMRV17} to
 the case of \SPTAs. 
The  idea behind this narrowing operation  is
that,  if one just observes
 $\Dirtree$,  uniform  $p$-labellings on 
$\Dirtree\times\Dirtreea$-trees  can be obtained by choosing the
labellings on $\Dirtree$-trees, and then lifting them to $\Dirtree\times\Dirtreea$-trees.

The construction and proof of correctness are straightforwardly
adapted from those
in~\cite{kupferman1999church} for \ATAs. 


\begin{thm}[Narrowing]
  \label{theo-narrow}
Given an \APTA $\SPTA$ on $\Dirtree\times\Dirtreea$-trees, one can build 
an  \APTA ${\narrow[\Dirtree]{\APTA}}$ on $\Dirtree$-trees such
  that for every pointed $(\APf,\Dirtree)$-tree $(\ltree,\noeud)$,
  every $\noeud'\in(\Dirtree\times\Dirtreea)^+$ such that $\projI[\Dirtree]{\noeud'}=\noeud$,
  and every initial stack content $\stackc_\init\in\stacka^*\cdot\stackb$, 
\[(\ltree,\noeud)\in\lang(\narrow[\Dirtree]{\APTA},\stackc_\init) \mbox{ iff
 }(\liftI[\Dirtree\times\Dirtreea]{}{\ltree},\noeud')\in\lang(\APTA,\stackc_\init).\]
\end{thm}

\begin{proof}
For a formula $\pform\in\boolp((\Dirtree\times\Dirtreea)\times \tQ\times\stacka^*)$, we let
$\projI[\Dirtree]{\pform}\,\in\boolp(\Dirtree\times \tQ\times\stacka^*)$ be the
formula obtained from $\pform$ by replacing each atom of the
form $[(\dir,\dira),\tq,\stackc]$ with atom $[\dir,\tq,\stackc]$.
We define the automaton
$\narrow[\Dirtree]{\APTA}=(\stacka,\tQ,\tdelta',\tq_{\init},\couleur)$ where for every $\tq\in\tQ$,
 $a\in 2^\APf$ and $\stacks\in\stacka$, $\tdelta'(\tq,a,\stacks)\egdef
 \projI[\Dirtree]{\tdelta(\tq,a,\stacks)}$. We now prove that this
 construction is correct.

Let $(\ltree,\noeud)$ be a pointed $(\APf,\Dirtree)$-tree, 
let  $\noeud'\in(\Dirtree\times\Dirtreea)^+$ be such that $\projI[\Dirtree]{\noeud'}=\noeud$,
and let $\stackc_\init\in\stacka^*\cdot\stackb$.
First, assume that $(\lift{}{\ltree},\noeud')\in\lang(\APTA,\stackc_\init)$. Let $\strat$ be a winning
 strategy for Eve in the acceptance game
$\tgame{\APTA}{\lift{}{\ltree}}{\noeud',\stackc_\init}$. 
By projecting on $\Dirtree$ nodes and formulas in positions of a
play $\fplay$ in this
game, \ie by replacing each position $(\noeuda,\tq,\stackc,\pform)$ with
$(\projI[\Dirtree]{\noeuda},\tq,\stackc,\projI[\Dirtree]{\pform})$, we obtain a
play $\projI[\Dirtree]{\fplay}$ in $\tgame{\narrow[\Dirtree]{\APTA}}{\ltree}{\noeud,\stackc_\init}$.
Applying this projection to the set of outcomes of $\strat$, we obtain
a set of plays $\text{Out}$ in
$\tgame{\narrow[\Dirtree]{\APTA}}{\ltree}{\noeud,\stackc_\init}$ that is the
set of all outcomes of some strategy $\strat'$ for Eve (there are actually
infinitely many such $\strat'$, which differ only on partial plays
that are not prefixes of plays in $\text{Out}$). And because
the sequences of states, and thus of colours, are the same in the projected and
original plays, $\strat'$ is winning for Eve in $\tgame{\narrow[\Dirtree]{\APTA}}{\ltree}{\noeud,\stackc_\init}$.

Now assume that $(\ltree,\noeud)\in\lang(\narrow[\Dirtree]{\APTA},\stackc_\init)$ and we
show that $(\lift{}{\ltree},\noeud')\in\lang(\APTA,\stackc_\init)$. 
There
exists a winning strategy $\strat$ for Eve in
$\game=\tgame{\narrow[\Dirtree]{\APTA}}{\ltree}{\noeud,\stackc_\init}$,
from which we define a winning strategy $\strat'$ for Eve in 
$\game'=\tgame{\APTA}{\lift{}{\ltree}}{\noeud',\stackc_\init}$.
Let $\fplay'$ be a partial play in $\game'$ in which it is Eve's turn
to play, \ie $\fplay'$ is of the form $\fplay''\cdot
(\noeuda',\tq',\stackc',\pform'_1\ou\pform'_2)$. Its projection on
$\Dirtree$ is thus of the form
$\projI[\Dirtree]{\fplay'}=\projI[\Dirtree]{\fplay''}\cdot
(\noeuda,\tq,\stackc,\pform_1\ou\pform_2)$,
where $\noeuda=\projI[\Dirtree]{\noeuda'}$, $\tq'=\tq$, $\stackc'=\stackc$ and $\pform_i=\projI[\Dirtree]{\pform'_i}$.
We let $\strat'(\fplay')\egdef(\noeuda',\tq',\stackc',\pform'_i)$, where $i$
is such that
$\strat(\projI[\Dirtree]{\fplay'})=(\noeuda,\tq,\stackc,\pform_i)$. 
Using the fact that a node $\noeuda'$ in $\lift{}{\ltree}$ is
labelled as $\projI[\Dirtree]{\noeuda'}$ in $\ltree$, one can check
that $\sigma'$ generates the same sequences of states of the automaton
as $\sigma$, and is thus winning for Eve.
\end{proof}


It then follows directly  that:

  \begin{prop}
    \label{prop-narrow-stable}
    If an \APTA $\APTA$ over $\Dirtree\times \Dirtreea \times \Dirtreeb$-trees is $\Dirtree$-guided, then so
    is ${\narrow[\Dirtree\times\Dirtreea]{\APTA}}$.
  \end{prop}

  Indeed the stack update function $\tdeltastack:
  2^{\APf}\times\stacka\times\Dirtree\to \stacka^*$ of $\APTA$ is also
  that of $\narrow[\Dirtree\times\Dirtreea]{\APTA}$.


  \section{Model checking hierarchical \QCTLsi}
\label{sec-proof}

{Before presenting our automata construction we introduce 
succinct unfoldings, which allow us to work with trees  over a
\emph{finite} set of directions.}

\subsection{Succinct unfoldings}
\label{sec-succinct-unfolding}

The semantics of
\QCTLsi on a finite \PCKS
$\PCKS=(\stacka,\setstates,\relation,\lab,\sstate_\init)$ is defined
via the  unfolding  of the infinite \CKS
$\CKS_\PCKS$, which is a tree
over the infinite set of directions $\setstates\times\stacka^*\cdot\stackb$.
In this tree each node contains the entire content of the
stack. It is however enough to record only the operations made on the
stack in each node: then, starting from the root and the initial stack
 $\stackb$, one can reconstruct the stack content at each node by
following the unique path from the root to this node, and applying the
successive operations on the stack. By doing so we obtain a tree over the finite set
of directions $\setstates\times\Dirstack$, where
\[\Dirstack=\{\stackc\mid
(\sstate,\stacks,\sstate',\stackc)\in\relation \text{ for some
}\sstate,\stacks \text{ and }\sstate'\}\union \{\epsilon\}\] (we require that  $\Dirstack$ always contain the
empty word).


First, for a partial path
$\spath=\conf{\sstate_\init}{\stackb}\conf{\sstate_1}{\stackc_{1}}\ldots\conf{\sstate_n}{\stackc_{n}}$
in $\PCKS$, we define its \defin{succinct representation}
\[\succpath(\spath)=\sconf{\sstate_\init}{\stackb}\sconf{\sstate_1}{\stackc'_{1}}\ldots\sconf{\sstate_n}{\stackc'_{n}}\in
(\setstates\times \Dirstack)^*\] where  for
$i\geq 0$, $\stackc'_{i+1}$ is such that
$\stackc_{i+1}=\stackc'_{i+1}\cdot \stackc''_i$, with
$\stackc_i=\stacks\cdot \stackc''_i$; that is, $\stackc'_{i+1}$ is
what has been pushed on the stack at step $i+1$. 

If
$\succpath=\sconf{\sstate_\init}{\stackb}\sconf{\sstate_1}{\stackc'_1}\ldots\sconf{\sstate_n}{\stackc'_n}\in
(\setstates\times \Dirstack)^*$ is a succinct representation, 
we can reconstruct the unique partial path
$\spath(\succpath)=\conf{\sstate_\init}{\stackb}\conf{\sstate_1}{\stackc_1}\ldots\conf{\sstate_n}{\stackc_n}$
 such that
$\succpath(\spath(\succpath))=\succpath$. 
We also let
$\stackcofsucc\egdef\stackc_n$ denote the stack
content after  $\succpath$. 


\begin{defi}
  \label{dec-succ-unfold}
Let $\PCKS=(\stacka,\setstates,\relation,\lab,\sstate_\init)$ 
be a \PCKS. 
Its \defin{succinct unfolding}  is the $(\APf,\setstates\times\Dirstack)$-tree
$\unfold[\PCKS]{\sstate,\stackc_0}\egdef (\tree,\lab')$ where
$\tree=\{\succpath(\spath)\mid \spath\in\PPaths(\PCKS)\}$ and for each
$\succpath\in\tree$ ending in $\sconf{\sstate_n}{\stackc'_n}$,
$\lab'(\succpath)=\lab(\sstate_n,\stackcofsucc)$.
\end{defi}

  The
following is a direct consequence of the semantics of \QCTLsi (recall
that $\phiprime$ is obtained from $\phi$ by replacing each concrete
observation $\cobs\subseteq [n]$ with
$\cobs'=\cobs\union\{n+1\}$):

\begin{lem}
  \label{lem-equivalence}
  For every \PCKS $\PCKS$ over
 $\{\setlstates_i\}_{i\in [n]}$ 
  and every \QCTLsi formula $\phi$, 
  $\PCKS\models\phi$ \,iff\, $\unfold[\PCKS]{\sstate_\init,\stackb}\models\phiprime$.
\end{lem}

{Note that succinct unfoldings were implicitly used in~\cite{aminof2013pushdown}.}



  
\subsection{Automata construction}
\label{sec-automata-construction}

We  generalise the automata construction
from~\cite{DBLP:conf/lics/BerthonMMRV17} to the case of pushdown
compound Kripke structures.  {The main novelties are, first, that we use
direction-guided pushdown tree automata instead of classic alternating
automata, relying on the fact that they are stable under the necessary
operations as proved in Section~\ref{sec-direction-guided}, and second, that  we  have to deal with regular
labellings for atomic propositions. 
 
 For the rest of this section we fix a \PCKS $\PCKS$ and a formula
 $\Phi\in \QCTLsih$. States in $\PCKS$ are elements of $\prodsetlstates$
 and concrete observations in $\Phi$ are subsets of $[n]$. But
 according to Lemma~\ref{lem-equivalence}, we will in fact
 consider the succinct unfolding of $\PCKS$ which is a tree over
 directions $\setstates\times\Dirstack$, where $\Dirstack=\setlstates_{n+1}$ captures
 stack operations,
 and with formula $\Phi_{n+1}$ in which each concrete observation
 $\cobs$ has been replaced with $\cobs\union\{n+1\}$, as the stack is visible.
More precisely, for each subformula $\phi$ of $\Phi_{n+1}$ we will build an automaton that works
on $\dirphi$-trees, where $\dirphi$ is defined as follows:
 
\begin{defi}
  \label{def-Iphi}
For every $\phi$, 
let $\Iphi\egdef
\biginter_{\cobs\in\setobs}\cobs$, where $\setobs$ is the set of
concrete observations
  that occur in $\phi$, with the intersection over the empty set 
  defined as $[n+1]$. We then let $\dirphi[\phi]\egdef
  \SI[\phi]\times\Dirstack$, where $\SI[\phi]=\{\projI[{\Iphi}]{\sstate}\mid\sstate\in\setstates\}$. 
\end{defi}

We assumed free atoms to be disjoint from quantified
ones in \QCTLsi formulas, \ie $\APq(\Phi)\inter\APfree(\Phi)=\emptyset$. We can thus assume that the \PCKS
$\PCKS$ is labelled  over $\APfree(\Phi)$, while the input trees of our
automata will be labelled over $\APq(\Phi)$.
To merge
the labels for quantified propositions carried by the (complete) input tree,
with those for free propositions carried by \PCKS $\CKS$, we use the
 merge operation from~\cite{DBLP:conf/lics/BerthonMMRV17}.

\begin{defi}
  \label{def-merge}
Let
 $\ltree=(\tree,\lab)$ be a complete
$(\APf,\Dirtree)$-tree and  $\ltree'=(\tree',\lab')$ an
$(\APf\,',\Dirtree)$-tree with same root as $\ltree$, where $\APf\inter\APf\,'=\emptyset$. 
The \defin{merge} of $\ltree$ and $\ltree'$
 is the $(\APf\union\APf\,',\Dirtree)$-tree\; \[\ltree\merge\ltree'\egdef
(\tree\cap\tree'=\tree',\lab''), \text{ where }\lab''(\noeud)=\lab(\noeud) \union \lab'(\noeud).\]
\end{defi}

We now describe our automata construction to inductively build automata
for subformulas of $\Phi_{n+1}$ and the fixed \PCKS
$\PCKS=(\stacka,\setstates,\relation,\labS,\sstate_\init)$.
{The construction is very similar to that
in~\cite{DBLP:conf/lics/BerthonMMRV17} for \QCTLsih on finite systems:
it builds on that for \CTLs~\cite{DBLP:journals/jacm/KupfermanVW00}, 
and in addition it uses narrowing, nondeterminisation and projection to
guess uniform labellings for quantified propositions. Also, in order not to
lose information on the  model while hiding components with the
narrowing operation, the model is encoded in the automata instead of
being given as input, and the input tree is only used
to  carry labellings for quantified propositions. 
 We only give here the cases
that contain significant differences from~\cite{DBLP:conf/lics/BerthonMMRV17}: {atomic proposition, and
  second-order quantification.}

{For atomic propositions  we
have to deal with regular labellings. To evaluate an atomic
proposition $p$ in a node $\noeud$ of the input tree, in state
$\sstate$ of $\PCKS$ and with stack content $\stackc$, automaton $\bigauto{p}$
will simply read the $p$-\labeling of $\noeud$ if $p$ is quantified;
otherwise it will simulate $\wautop{\sstate}$  (which represents the
regular labelling for $p$ in state $\sstate$) on the stack content by popping symbol after
symbol, feeding them to $\wautop{\sstate}$ while following an
arbitrary direction in the input tree, and it will accept if
$\wautop{\sstate}$ is in an accepting state when the stack is empty.}

{For $\phi=\existsp{\cobs}\phi'$, we will first use the
  induction hypothesis to build a $\Dirstack$-guided automaton for $\phi'$. Then we
  will use the results established in
  Section~\ref{sec-direction-guided} to, first, narrow it to make it
  observe only $\cobs$, then nondeterminise it, which is possible
  because $\Dirstack$-guided automata are semi-alternating, and finally project it
  over $p$, while remaining $\Dirstack$-guided.}



\begin{lem}
    \label{lem-final}
    For every subformula
    $\phi$ of $\Phi_{n+1}$ and state $\sstate\in\setstates$,  one can build an
    \APTA $\bigauto[\sstate]{\phi}$ on
    $(\APq(\Phi),\dirphi)$-trees with $\Dirstack$-guided stack
and    such that, for every 
    $(\APq(\Phi),\dirphi)$-tree $\ltree$ rooted in
    $(\projI[{\Iphi}]{\sstate_\init},\stackb)$ and every 
    {partial path $\spath\in\PPaths(\PCKS)$ ending in
    $\conf{\sstate}{\stackc}$,}  
    it holds that
    \[
      (\ltree,\projI[{\Iphi}]{{\succpath(\spath)}})\in\lang(\bigauto[\sstate]{\phi},{\stackc})
      \mbox{\;\;\;iff\;\;\;}
      \liftI[{\setstates\times\Dirstack}]{}{\ltree}\merge\;\unfold{\sstate,\stackc},{\succpath(\spath)}
      \modelst \phi. \]
  
  \end{lem}

  \begin{proof}
        The proof is by induction on $\phi$.
The automata we build will be $\Dirstack$-guided, and more precisely
the stack update
    function $\tdeltastack: 2^{\APf}\times\stacka\times\Dirstack\to \stacka^*$
 will be the mapping $(a,\stacks,\stackc)\mapsto \stackc$.

    $\bm{\phi=p:}$ Recall that
    automaton $\wautop{\sstate}$  accepts the language
    $\{\stackc\in\stacka^*\cdot\stackb\mid p\in\val(\cstate,\stackc)\}$.
From $\wautop{\sstate}$ we can easily build an \APTA $\APTA'$
over $(\setstates\times\Dirstack)$-trees (by
Definition~\ref{def-Iphi}, $\dirphi=\setstates\times\Dirstack$) that,
when started in a node $\noeud$ with a stack content $\stackc$,
simulates $\wautop{\sstate}$ on the stack by popping it until reaching
$\stackb$,  going down in the input tree in direction
$(\sstate,\epsilon)$ for some arbitrary state
$\sstate\in\setstates$. Automaton $\APTA'$ accepts if  $\wautop{\sstate}$ is
in an accepting state when the bottom of the stack is reached.

Writing $\APTA'=(\stacka,\tQ,\tdelta,\tq_\init,\couleur)$, we 
define \[\bigauto{p}=(\stacka,\tQ\union\{\tq'_\init\},\tdelta',\tq'_\init,\couleur'),\]
where $\tq'_\init$ is a fresh initial state, $\couleur'$ extends
$\couleur$ by assigning some insignificant colour to $\tq'_\init$, and
$\tdelta'$ extends $\tdelta$ by letting, for each $\stacks\in\stacka$
and $a\in 2^\APq$,
\[\tdelta'(\tq'_\init,\stacks,a)=
  \begin{cases}
    \top & \mbox{if }p\in \APq \mbox{ and }p\in a\\
    \perp & \mbox{if }p\in \APq \mbox{ and }p\notin a\\
    \tdelta(\tq_\init,\stacks,a) & \mbox{otherwise, \ie if }p\in \APfree
  \end{cases}
\]

    
$\bm{\phi=\neg\phi':}$
We obtain $\bigauto{\phi}$  by
   complementing $\bigauto{\phi'}$.

$\bm{\phi=\phi_{1}\ou\phi_{2}:}$ We first build
  $\bigauto{\phi_1}$ and $\bigauto{\phi_2}$. Each $\bigauto{\phi_{i}}$
  works on $\dirphi[\phi_{i}]$-trees, and
  $\Iphi=\Iphi[\phi_{1}]\cap\Iphi[\phi_{2}]$, so that by definition
  $\dirphi=\SI[{\Iphi[\phi_1]\inter\Iphi[\phi_2]\inter[n]}]\times\Dirstack$. Thus
  we first narrow down each $\bigauto{\phi_{i}}$ so that they both
  work on $\dirphi$-trees: 
  for $i\in \{1,2\}$, we let $\APTA_{i}\egdef
  {\narrow[\Iphi]{\bigauto{\phi_{i}}}}=(\stacka,\tQ^{i},\tdelta^{i},\tq_\init^{i},\couleur^{i})$.
  %
  Letting
  $\tq_\init$ be a fresh initial state we define
  $\bigauto{\phi}\egdef(\stacka,\{\tq_\init\}\union\tQ^1\union\tQ^2,\tdelta,\tq_\init,\couleur)$, where
  $\tdelta$ and $\couleur$ agree with $\tdelta^i$ and $\couleur^i$,
  respectively, on states from $\tQ^i$,  and
  $\tdelta(\tq_\init,\stacks,a)=\tdelta^1(\tq_\init^1,\stacks,a)\ou
  \tdelta^2(\tq_\init^2,\stacks,a)$. The colour of $\tq_\init$ does not matter.

$\bm{\phi=\E\psi:}$
  Let $\max(\psi)=\{\phi_1,\ldots,\phi_k\}$ be the
  set of maximal state sub-formulas of $\psi$.
  In a first step we
  see these maximal state sub-formulas as atomic propositions, we see the formula
  $\psi$ as an \LTL formula, and we  build 
  a nondeterministic
parity word automaton
  $\autopsi=(\Qpsi,\Deltapsi,\qpsi_\init,\couleurpsi)$ over alphabet $2^{\max(\psi)}$
  that accepts exactly the models of $\psi$~\cite{vardi1994reasoning}.\footnote{Note
    that, as usual for nondeterministic word automata, we take the
    transition function of type
    $\Deltapsi:\Qpsi\times 2^{\max(\psi)}\to
    2^{\Qpsi}$. {Note also that these automata use two
      colours, as actually 
    B\"uchi automata are enough.}}  We define
  the \APTA $\tauto$ that, given as input a
  $(\max(\psi),\dirphi)$-tree $\ltree$, nondeterministically
  guesses a path $\tpath$ in
  $\unfold{\sstate,\stackc}$ and simulates
  $\autopsi$ on it, assuming that the labels it reads while following
  $\projI[\Iphi]{\tpath}$ in its input correctly represent the truth
  value of formulas in $\max(\psi)$ along
  $\tpath$. 
We define
 $\tauto\egdef(\stacka,\tQ,\tdelta,\tq_{\init},\couleur)$, where
\begin{itemize}
\item $\tQ=\Qpsi\times\setstates$, 
\item $\tq_{\init}=(\qpsi_{\init},\sstate)$,
\item for each $(\qpsi,\sstate')\in\tQ$, $\couleur(\qpsi,\sstate')=\couleurpsi(\qpsi)$, and
\item  for each $(\qpsi,\sstate')\in\tQ$
  and $a\in 2^{\max(\psi)}$, 
  \[\tdelta((\qpsi,\sstate'),\stacks,a)=\bigvee_{\tq'\in\Deltapsi(\qpsi,a)}\quad\bigvee_{
    (\sstate'',\stackc) \text{ s.t. } (\sstate',\stacks,\sstate'',\stackc)\in\relation}[(\projI[\Iphi]{\sstate''},\stackc),\left(\tq',\sstate''\right),\stackc].\]
\end{itemize}
Intuitively, $\tauto$ reads the current label $a$ in its
input and the top symbol of
the stack $\stacks$. It then chooses nondeterministically
which transition to take in $\autopsi$, and it chooses also a possible
transition $(\sstate',\stacks,\sstate'',\stackc)$ in $\PCKS$. Then it moves
in the input tree in direction $(\projI[\Iphi]{\sstate''},\stackc)$, sending
there a state that records the new current state in $\autopsi$ and the
new current state in $\PCKS$, and it pushes $\stackc$ on the stack of
the automaton.

In general it is not possible to define a $\max(\psi)$-labelling
  of $\ltree$ that faithfully represents the truth values of formulas
  in $\max(\psi)$, because a node in
$\ltree$ may correspond to  different nodes in
$\unfold{}$ that have same projection on $\dirphi$ but satisfy different formulas of $\max(\psi)$. However this is
not a problem because different copies of the final automaton (defined
below) that
visit the same node can guess different
labellings, depending on the actual state of $\PCKS$.

From $\tauto$ we build automaton $\bigauto{\phi}$ over
$\dirphi$-trees labelled with  atomic propositions in
$\APq$.
In each node it visits, $\bigauto{\phi}$ guesses which formulas of
$\max(\psi)$ hold in this node  with the current state of $\PCKS$ and current stack content, it simulates $\tauto$
accordingly, and checks
that the guess it made is correct.
If the path being guessed in $\unfold{\sstate,\stackc}$
is currently in node $\succpath$ ending with state $\sstate'$ and stack
content $\stackcofsucc$, and
$\bigauto{\phi}$ guesses that $\phi_{i}$ holds,
it launches a
copy of automaton $\bigauto[\sstate']{\phi_{i}}$ from node
 $\projI[\Iphi]{\succpath}$ in its input $\ltree$, with the
 current stack content $\stackcofsucc$.

 For each $\sstate'\in\setstates$ state of $\PCKS$, and each
 $\phi_{i}\in\max(\psi)$, we first build $\bigauto[\sstate']{\phi_i}$
 which works on $\dirphi[\phi_i]$-trees.  We narrow down these
 automata to work on $\Iphi[\phi]=\inter_{i=1}^k \Iphi[\phi_i]$: let
 $\APTA^i_{\sstate'}\egdef\narrow[{\Iphi[\phi]}]{\bigauto[\sstate']{\phi_i}}
 =(\stacka,\tQ^{i}_{\sstate'},\tdelta^{i}_{\sstate'},\tq^{i}_{\sstate'},\couleur^{i}_{\sstate'})$.
 We also let
 $\compl{\APTA^{i}_{\sstate'}}=(\stacka,\compl{\tQ^{i}_{\sstate'}},\compl{\delta^{i}_{\sstate'}},\compl{\tq^{i}_{\sstate'}},\compl{\couleur^{i}_{\sstate'}})$
 be the dualisation of $\APTA^i_{\sstate'}$, and we assume that all the state sets are pairwise disjoint.
We define the \APTA
\[\bigauto{\phi}=(\stacka,\tQ\cup
\bigcup_{i,\sstate'} \tQ^{i}_{\sstate'} \cup
\compl{\tQ^{i}_{\sstate'}},\tdelta',\tq_{\init},\couleur'),\] where the
colours of states remain unchanged,
and $\tdelta'$ is defined as follows. For states in $\tQ^{i}_{\sstate'}$
(resp. $\compl{\tQ^{i}_{\sstate'}}$), $\tdelta'$ agrees with $\tdelta^{i}_{\sstate'}$
(resp. $\compl{\delta^{i}_{\sstate'}}$), and for $(\qpsi,\sstate')\in
\tQ$, $\stacks\in\stacka$ and $a\in
2^{\APq}$ we let 
\begin{align*}
\tdelta'((\qpsi,\sstate'),\stacks,a)=  \bigou_{a'\in
     2^{\max(\psi)}}& 
    \Bigg ( \tdelta\left((\qpsi,\sstate'),\stacks,a'\right)
    \et 
      \biget_{\phi_i\in a'}\tdelta^{i}_{\sstate'}(\tq^{i}_{\sstate'},\stacks,a) 
    \;\et
     \biget_{\phi_i\notin
      a'}\compl{\delta^{i}_{\sstate'}}(\compl{\tq^{i}_{\sstate'}},\stacks,a)
    \Bigg ).
\end{align*}

$\bm{\phi=\exists}^{\bm{\cobs}}\bm{p.\,\phi':}$
First, we build automaton $\bigauto{\phi'}$ that works on $\dirphi[\phi']$-trees;
since $\phi$ is hierarchical, we have that $\Iphi=\cobs\subseteq
\Iphi[\phi']$ 
and we can narrow down $\bigauto{\phi'}$ to work on $\dirphi$-trees:
we let
$\APTA_{1}\egdef{\narrow[{\cobs}]{\bigauto{\phi'}}}$. 
By induction hypothesis, $\bigauto{\phi'}$ is $\Dirstack$-guided.
By definition of $\Phi_{n+1}$ we have
$n+1\in\cobs$, and thus   $\APTA_1$ is also $\Dirstack$-guided, by
Proposition~\ref{prop-narrow-stable}. 
Now, by
Theorem~\ref{theo-simulation} we can
nondeterminise $\APTA_1$, getting $\APTA_{2}$, which by
Theorem~\ref{theo-projection} we can project with respect to
$p$, obtaining $\bigauto{\phi}\egdef \proj{\APTA_{2}}$.

\subsubsection*{Correctness}
In the following, for  $J\subseteq I \subseteq [n+1]$, for every $(\APf,\Dirtreei[J])$-tree $\ltree$ with root
 $\dir\in\Dirtreei[J]$, and every  $\dira\in\Dirtreei[I\setminus J]$, we note
$\liftI{\dira}{\ltree}$ for $\liftI[{\Dirtreei[{I}]}]{\dira}{\ltree}$
(recall that $\Dirtreei =\prod_{i\in
I}\setlstates_{i}$, and that $\setlstates_{n+1}=\Dirstack$).

Let $\ltree=(\tree,\lab)$ be a complete
$(\APq(\Phi),\dirphi)$-tree rooted in
$(\projI[{\Iphi}]{\sstate},\stackb)$, let
$\spath\in\PPaths(\PCKS)$ be some partial path ending in
    $\conf{\sstate}{\stackc}$, and let $\succpath=\succpath(\spath)$.


$\bm{\phi=p:}$  First, note that $I_{p}=[n+1]$, so that
 $\projI[{\Iphi}]{\sstate}=\sstate$,  $\ltree$ is
 rooted in $(\sstate,\stackb)$,
 $\projI[{\Iphi[p]}]{\succpath}=\succpath$ and $\liftI[{[n+1]}]{}{\ltree}=\ltree$.
Let us consider first
the case where $p\in\APfree$: by definition of $\bigauto{p}$, we have
$(\ltree,\succpath)\in\lang(\bigauto{p},\stackc)$ iff $\stackc$ is accepted by
$\wauto^p_\sstate$, \ie iff $p\in \labS(\conf{\sstate}{\stackc})$;
also, by definition of the merge, 
we have that
      $\ltree\merge\;\unfold{\sstate,\stackcofsucc},\succpath\models  p$
 iff
 $p\in\labS(\conf{\sstate}{\stackc})$, which concludes this case.
Now if
 $p\in\APq$, by definition of $\bigauto{p}$, we have that
 $(\ltree,\succpath)\in\lang(\bigauto{p},\stackc)$
 iff node $\succpath$ is
 labelled with $p$ in $\ltree$. On the other
 hand, by definition  of the merge, we have
$\ltree\merge\;\unfold{\sstate,\stackc},\succpath\models  p$ iff
$\succpath$ is labelled with $p$ in $\ltree$, and we are done. 
 
 $\bm{\phi=\neg\phi':}$ 
 This case is trivial.
We only remark that the dualisation of a
   $\Dirstack$-guided \APTA
 is also $\Dirstack$-guided.
 
 $\bm{\phi_{1}\ou \phi_{2}:}$
For $i\in\{1,2\}$  we have $\APTA_{i} =
\narrow[{\Iphi}]{\bigauto{\phi_{i}}}$, so by Theorem~\ref{theo-narrow}
 we have
 that \[(\ltree,\projI[{\Iphi}]{\succpath})\in\lang(\APTA_{i},\stackc)
   \; \text{ iff } \;
(\liftI[{\Iphi[\phi_i]}]{}{\ltree},\projI[{\Iphi[\phi_i]}]{\succpath})\in\lang(\bigauto{\phi_{i}},\stackc).\] By
induction hypothesis the latter holds iff
\[\liftI[{[n+1]}]{}{\liftI[{\Iphi[\phi_{i}]}]{}{\ltree}\,}\merge\;
\unfold{\sstate,\stackcofsucc},\succpath \modelst \phi_{i},\] and thus
 \[(\ltree,\projI[{\Iphi}]{\succpath})\in\lang(\APTA_{i},\stackc)  \;\text{iff}\;
\liftI[{[n+1]}]{}{\ltree}\merge\;\unfold{\sstate,\stackcofsucc},\succpath\modelst\phi_{i}.\]
We conclude by noting that
$\lang(\bigauto{\phi},\stackc)=\lang(\APTA_{1},\stackc)\union\lang(\APTA_{2},\stackc)$.

$\bm{\phi=\E\psi :}$ Suppose that
$\liftI[{[n+1]}]{}{\ltree}\merge\;\unfold{\sstate,\stackcofsucc},\succpath\modelst\E\psi$. By
definition, there
exists an infinite path $\tpath'$ that starts at node $\succpath$ of
$\liftI[{[n+1]}]{}{\ltree}\merge\;\unfold{\sstate,\stackcofsucc}$ such
that
$\liftI[{[n+1]}]{}{\ltree}\merge\;\unfold{\sstate,\stackcofsucc},\tpath'\models\psi$,
{and by definition of  succinct unfoldings and  merge operation, $\tpath'$
corresponds to a unique infinite path $\spathinf\in\Paths(\PCKS)$ that extends
$\spath$.}  Again, let
$\max(\psi)$ be the set of maximal state subformulas of
$\phi$, and let
$w(\tpath')$ be the infinite word over
$2^{\max(\psi)}$ that agrees with
$\tpath'$ on the state formulas in
$\max(\psi)$, \ie for each node $\tpath'_k$ of
$\tpath'$ and formula
$\phi_i\in\max(\psi)$, it holds that
\[\phi_i\in w(\tpath')_k \; \text{ iff }\;
  \liftI[{[n+1]}]{}{\ltree}\merge\;\unfold{\sstate,\stackcofsucc},\tpath'_k
  \modelst \phi_i.\] By definition,
$\autopsi$ has an accepting execution on
$w(\tpath')$. To show that
$(\ltree,\projI[\Iphi]{\succpath})\in\lang(\bigauto[\sstate]{\phi},\stackc)$
we show that Eve can win the acceptance game
$\tgame{\bigauto[\sstate]{\phi}}{\ltree}{\projI[\Iphi]{\succpath},\stackc}$.
In this game, Eve can guess {the continuation $\spathinf$ of
  $\spath$, or equivalently the path $\tpath'$ in $\liftI[{[n+1]}]{}{\ltree}\merge\;\unfold{\sstate,\stackcofsucc}$,}
 while the automaton follows
$\tpath=\projI[{\Iphi[\phi]}]{\tpath'}$ in its input
$\ltree$, and she can also guess the corresponding word
$w(\tpath')$ on
$2^{\max(\psi)}$ and an accepting execution of
$\autopsi$ on
$w(\tpath')$. Let
$\succpath'\in\liftI[{[n+1]}]{}{\ltree}\merge\;\unfold{\sstate,\stackc}$
be a node along $\tpath'$, let
$(\sstate',\stackc')$ be its last direction and let
$\succpath''=\projI[{\Iphi[\phi]}]{\succpath'}\in\ltree$. Assume that
in node
$\succpath''$ of the input tree, in a state $(\qpsi,\sstate')\in
\tQ$, Adam challenges Eve on some
$\phi_i\in\max(\psi)$ that she assumes to be true in
$\succpath'$, \ie Adam chooses the conjunct
$\tdelta^{i}_{\sstate'}(\tq^{i}_{\sstate'},\stacks,a)$, where
$\stacks$ is the top of the current stack content and
$a$ is the label of
$\succpath''$. Note that in the evaluation game this means that Adam
moves to position
$(\succpath'',(\qpsi,\sstate'),\stackcofsucc[\succpath'],\tdelta^{i}_{\sstate'}(\tq^{i}_{\sstate'},\stacks,a))$.
We want to show that Eve wins from this position.  To do so we first
show that
$(\ltree,\succpath'')\in\lang(\APTA^{i}_{\sstate'},\stackcofsucc[\succpath'])$.

First, recall
    that     $\APTA^{i}_{\sstate'}=\narrow[I_{\phi}]{\bigauto[\sstate']{\phi_{i}}}$.
    By Theorem~\ref{theo-narrow}, it holds that
    $(\ltree,\succpath'')\in\lang(\APTA^{i}_{\sstate'},\stackcofsucc[\succpath'])$
    iff
    $(\liftI[{\Iphi[\phi_i]}]{}{\ltree},\projI[{\Iphi[\phi_i]}]{\succpath'})\in\lang(\bigauto[\sstate']{\phi_{i}},\stackcofsucc[\succpath'])$.
    Next, by applying the induction hypothesis we get that
    \[(\liftI[{\Iphi[\phi_i]}]{}{\ltree},\projI[{\Iphi[\phi_i]}]{\succpath'})\in\lang(\bigauto[\sstate']{\phi_{i}},\stackcofsucc[\succpath'])\;
    \text{ iff }\;
    \liftI[{[n+1]}]{}{\liftI[{\Iphi[\phi_i]}]{}{\ltree}}\merge\;
    \unfold{},\succpath'\modelst \phi_{i},\] \ie iff
    $\liftI[{[n+1]}]{}{\ltree}\merge\; \unfold{},\succpath'\modelst
    \phi_{i}$, which holds because we have assumed that Eve guesses
    $w$ correctly.

Eve thus has a winning strategy from the initial position of
$\tgame{\APTA^{i}_{\sstate'}}{\ltree}{\succpath'',\stackcofsucc[\succpath']}$,
the acceptance game of $\APTA^{i}_{\sstate'}$ on $(\ltree,\succpath'')$
with initial stack content $\stackcofsucc[\succpath']$.  This initial
position is
\[(\succpath'',\tq^i_{\sstate'},\stackcofsucc[\succpath'],\tdelta^i_{\sstate'}(\tq^i_{\sstate'},\stacks,a)).\]
Since
this position 
and position
\[(\succpath'',(\qpsi,\sstate'),\stackcofsucc[\succpath'],\tdelta^i_{\sstate'}(\tq^i_{\sstate'},\stacks,a))\]
in $\tgame{\bigauto[\sstate]{\phi}}{\ltree}{\projI[\Iphi]{\succpath},\stackcofsucc}$
contain the same node $\succpath'$, stack content $\stackcofsucc[\succpath']$ and transition formula
$\tdelta^i_{\sstate'}(\tq^i_{\sstate'},\stacks,a)$, a winning strategy in one
of these positions\footnote{Recall that positional strategies are
  sufficient in parity games \cite{DBLP:journals/tcs/Zielonka98}.} is
also a winning strategy in the other, and therefore Eve wins Adam's
challenge.  With a similar argument, we get that also when Adam
challenges Eve on some $\phi_i$ assumed not to be true in node
$\succpath'$, Eve wins the challenge, which concludes this direction. 

    For the other direction, assume that
    $(\ltree,\projI[\Iphi]{\succpath})\in\lang(\bigauto{\phi},\stackc)$, \ie Eve wins the evaluation game
      $\tgame{\bigauto[\sstate]{\phi}}{\ltree}{\projI[\Iphi]{\succpath},\stackc}$.  A winning strategy for Eve
    describes a path $\tpath$ 
    in $\unfold{\sstate}$ starting in node $\succpath$, which is also
 a path in $\liftI[{[n+1]}]{}{\ltree}\merge\;\unfold{\sstate,\stackcofsucc}$. This
    winning strategy also defines an infinite word $w(\tpath)$
    over $2^{\max(\psi)}$ such that $w(\tpath)$ agrees with $\tpath$ on the
    formulas in $\max(\psi)$, and it also describes an accepting run
    of $\autopsi$ on $w$. Hence $\liftI[{[n+1]}]{}{\ltree}\merge\;\unfold{\sstate,\stackcofsucc},\tpath\modelst\psi$, and
    $\liftI[{[n+1]}]{}{\ltree}\merge\;\unfold{},\succpath\modelst
    \phi$.

$\bm{\phi=\exists}^{\bm{\cobs}}\bm{p.\,\phi':}$
    First, by definition we have $\Iphi=\cobs\inter\Iphi[\phi']$. Because
    $\phi$ is hierarchical we have that $\cobs\subseteq\cobs'$ for
    every $\cobs'$ that occurs in $\phi'$, and thus $\cobs\subseteq\Iphi[\phi']$.
    It follows that $\Iphi=\cobs$.

    Next, since   $\bigauto{\phi}= \proj{\APTA_{2}}$, by Theorem~\ref{theo-projection} we have that
    \begin{align*}
      \label{eq:5}
      (\ltree,\projI[\Iphi]{\succpath})\in\lang(\bigauto{\phi},\stackcofsucc)
      \mbox{\bigiff} 
      \exists\,\plab \mbox{ a $p$-\labeling for
    $\ltree$ such that 
  } (\ltree\prodlab\plab,\projI[\Iphi]{\succpath})\in\lang(\ATA_{2},\stackc).        
    \end{align*}
By  Theorem~\ref{theo-simulation} for simulation,
    $\lang(\ATA_{2},\stackc)=\lang(\ATA_{1},\stackc)$, and since $\ATA_{1}
    =\narrow[\cobs]{\bigauto{\phi'}}=\narrow[{\Iphi[\phi]}]{\bigauto{\phi'}}$
    we get by Theorem~\ref{theo-narrow} that
    \begin{equation*}
      \label{eq:6}
      (\ltree\prodlab\plab,\projI[\Iphi]{\succpath})\in\lang(\ATA_{2}) \mbox{\bigiff}
      (\liftI[{\Iphi[\phi']}]{}{(\ltree\prodlab\plab)},\projI[{\Iphi[\phi']}]{\succpath})\in\lang(\bigauto{\phi'}).   
    \end{equation*}
By induction hypothesis, 
\begin{align*}
(\liftI[{\Iphi[\phi']}]{}{(\ltree\prodlab\plab)},\projI[{\Iphi[\phi']}]{\succpath})\in\lang(\bigauto{\phi'})
 \mbox{\bigiff} \liftI[{[n+1]}]{}{\liftI[{\Iphi[\phi']}]{}{(\ltree\prodlab\plab)}}\merge\;\unfold{\sstate},\succpath\modelst
\phi'.  
\end{align*}
The three equivalences above 
plus the fact that
 $\liftI[{[n+1]}]{}{\liftI[{\Iphi[\phi']}]{}{(\ltree\prodlab\plab)}}=   \liftI[{[n+1]}]{}{(\ltree\prodlab\plab)}$
imply that
\[\begin{array}{c}
   (\ltree,\projI[\Iphi]{\succpath})\in\lang(\bigauto{\phi},\stackc)
  \\ \mbox{\bigiff}\\ 
        \exists\,\plab \mbox{ a $p$-\labeling for
    $\ltree$ s.t.
  }
  \liftI[{[n+1]}]{}{(\ltree\prodlab\plab)}\merge\;\unfold{\sstate},\succpath\modelst\phi'.
\end{array}\]

We now prove the following which, together with the latter equivalence,
concludes the proof:

\begin{equation}
  \label{eq:9}
  \begin{array}{c}
      \exists\,\plab \mbox{ a $p$-\labeling for
    $\ltree$ s.t.
  }
  \liftI[{[n+1]}]{}{(\ltree\prodlab\plab)}\merge\;\unfold{\sstate},\succpath\modelst\phi'
  \\
    \mbox{\bigiff}\\
  \liftI[{[n+1]}]{}{\ltree}\merge\;\unfold{\sstate},\succpath\modelst\existsp{\cobs}\phi'
  \end{array}
\end{equation}

Assume that there exists a $p$-\labeling $\plab$ for $\ltree$ such
that
$\liftI[{[n+1]}]{}{(\ltree\prodlab\plab)}\merge\;\unfold{\sstate},\succpath\modelst\phi'$.
Let $\plab'$ be the $p$-\labeling of
$\liftI[{[n+1]}]{}{(\ltree\prodlab\plab)}\merge\;\unfold{\sstate}$.
By definition of the merge, $\plab'$ is equal to the $p$-labelling of
$\liftI[{[n+1]}]{}{(\ltree\prodlab\plab)}$; therefore
\[\liftI[{[n+1]}]{}{(\ltree\prodlab\plab)}\merge\;\unfold{\sstate}=(\liftI[{[n+1]}]{}{\ltree}\merge\;\unfold{\sstate})\prodlab\plab',\]
and $\plab'$ is $\Iphi$-uniform, \ie  $\cobs$-uniform (by definition of the
widening).
This concludes this direction.

Now assume that
$\liftI[{[n+1]}]{}{\ltree}\merge\;\unfold{\sstate},\succpath\modelst\existsp{\cobs}\phi'$:
there exists a $\cobs$-uniform  $p$-\labeling $\plab'$ for $\liftI[{[n+1]}]{}{\ltree}\merge\;\unfold{\sstate}$ 
 such that $(\liftI[{[n+1]}]{}{\ltree}\merge\;\unfold{\sstate})\prodlab\plab',\noeud\modelst\phi'$.  We define
 a $p$-\labeling $\plab$ for $\ltree$ such that
 $\liftI[{[n+1]}]{}{(\ltree\prodlab\plab)}\merge\;\unfold{\sstate},\succpath\modelst\phi'$.
 First, let us write $\ltree'=\liftI[{[n+1]}]{}{\ltree}\merge\;\unfold{\sstate}=(\tree',\lab')$.
 For each node $\noeud$ of $\ltree$, 
 let
 \[
\plab(\noeud)=
\begin{cases}
\plab'(\noeud') &   \mbox{if there exists }\noeud'\in\tree' \mbox{
  such that }\projI[\cobs]{\noeud'}=\noeud,\\
0 & \mbox{otherwise.}
\end{cases}
   \]
   This is well defined because $\plab'$ is $\cobs$-uniform in $p$, so
   that if two nodes $\noeud',\noeuda'$ project on $\noeud$, \ie
   $\noeud'\oequivt\noeuda'$, we have
   $\plab'(\noeud')=\plab'(\noeuda')$.  In case there is no
   $\noeud'\in\tree'$ such that
   $\projI[{\Iphi[\phi]}]{\noeud'}=\noeud$, the value of
   $\plab(\noeud)$ has no impact on
   $\liftI[{[n+1]}]{}{(\ltree\prodlab\plab)}\merge\;\unfold{\sstate}$.
Finally,
   $\liftI[{[n+1]}]{}{(\ltree\prodlab\plab)}\merge\;\unfold{\sstate}=
   (\liftI[{[n+1]}]{}{\ltree}\merge\;\unfold{\sstate})\prodlab\plab'$,
hence the result.
 
\end{proof}

\subsection{Proof of Theorem~\ref{theo-decidable-QCTLi}}
\label{sec-proof-theo-decidable}

We  now prove Theorem~\ref{theo-decidable-QCTLi}. Let
$\PCKS$ be a \PCKS with initial state $\sstate_\init$ and $\Phi\in\QCTLsih$.
For readability let us also write $\Phi'=\Phi_{n+1}$. Applying Lemma~\ref{lem-final} to $\Phi'$ and state
$\sstate_\init$, we can construct an \APTA
$\bigauto[\sstate_\init]{\Phi'}$  with $\Dirstack$-guided stack
    such that for every 
    $(\APq(\Phi),\dirphi[\Phi'])$-tree $\ltree$ rooted in
    $(\projI[{\Iphi[\Phi']}]{\sstate_\init},\stackb)$, every 
 partial path $\spath$  in $\PPaths(\PCKS)$ ending in
    $\conf{\sstate_\init}{\stackc}$, 
    it holds that
      \[(\ltree,\projI[{\Iphi[\Phi']}]{\succpath(\spath)})\in\lang(\bigauto[\sstate_\init]{\Phi'},\stackc)
      \mbox{ iff }
      \liftI[{\setstates\times\Dirstack}]{}{\ltree}\merge\;\unfold{},\succpath(\spath) \modelst
      \Phi'.\] 

      Let $\tree$ be the full $\dirphi[\Phi']$-tree rooted in
      $(\projI[{\Iphi[\Phi']}]{\sstate_\init},\stackb)$, and let
      $\ltree=(\tree,\lab_{\emptyset})$, where $\lab_{\emptyset}$ is
       the empty labelling.
      {Clearly, we have that}
$\liftI[{\setstates\times\Dirstack}]{}{\ltree}\merge
\;\unfold{}=\unfold{}$, and because $\ltree$ is
rooted in $(\projI[{\Iphi[\Phi']}]{\sstate_\init},\stackb)$, applying
the above equivalence to $\ltree$ and  $\spath=\conf{\sstate_\init}{\stackb}$, we get that 
      $(\ltree,(\projI[{\Iphi[\Phi']}]{\sstate_\init},\stackb))\in\lang(\bigauto[\sstate_\init]{\Phi'},\stackb) \mbox{ iff }
        \unfold{\sstate_\init}\modelst \Phi'$.

        Since, by Lemma~\ref{lem-equivalence},
        $\unfold{\sstate_\init}\modelst \Phi'$ holds iff
        $\PCKS\modelst\Phi$, it only remains to check whether tree
        $\ltree$, which is regular\footnote{{A tree is regular if
          it has only finitely many distinct infinite subtrees;
          equivalently if it can be obtained by unfolding a finite
          labelled Kripke structure.}}, is accepted by
        $\bigauto[\sstate_\init]{\Phi'}$. This can be done by taking the product of
        $\bigauto[\sstate_\init]{\Phi'}$ with a finite Kripke
        structure representing $\ltree$ and checking for emptiness,
        which is decidable for semi-alternating pushdown tree
        automata~\cite{aminof2013pushdown}.

\section{\SL with imperfect information}
\label{sec-sl}

We recall the syntax and semantics of Strategy Logic with imperfect
information (\SLi). 
The semantics is
defined as in~\cite{DBLP:conf/lics/BerthonMMRV17} on concurrent game arenas with imperfect
information, except that we allow for infinite ones. We then define
the subclass of infinite arenas generated by pushdown arenas with
imperfect information on control states, on which we study the
model-checking problem for \SLi.


\subsection{Syntax}
\label{sec-sl-syntax}
 
For the rest of the section we fix  a finite  set of
 \defin{agents} or
\defin{players} $\Agf$, a finite  set of \defin{observation symbols}
or simply \defin{observations}
$\Obsf$ and
a finite set of \defin{variables} $\Varf$. Observations
represent observational powers for the players. 

\begin{defi}
  \label{def-SLi}
    The syntax of \SLi is defined by the following grammar:
    \begin{align*}
 \phi\egdef &\; p 
  \mid \neg \phi 
  \mid \phi\ou\phi 
  \mid \Estrato{\obs}\phi 
               \mid \bind{\var}\phi
  \mid \Eout\psi & \text{State formulas }               
      \\
      \psi\egdef &\; \phi
                   \mid \neg \psi
                   \mid \psi\ou \psi
                   \mid \X \psi
                   \mid  \psi \until \psi & \text{Path formulas } 
    \end{align*}
     where 
  $p\in\APf$, $\var\in\Varf$, $\obs\in\Obsf$ and $a\in\Agf$.
\end{defi}

Boolean and temporal operators have their usual meaning. Strategy quantification $\Estrato{\obs}\phi$ reads as ``there exists a strategy $\var$
that takes decisions based on observational power $\obs$  such that $\phi$
holds''.
Binding $\bind{\var}\phi$ reads as ``when agent $\ag$ plays strategy $\var$,
 $\phi$ holds'', 
and finally, $\Eout\psi$ reads as ``$\psi$
holds in some
outcome of the strategies currently used by the players''.
 \SLi consists of all state formulas.

For $\phi\in\SLi$, we let  $\free(\phi)$ be the set of variables that appear
free in $\phi$, \ie that
appear out of the scope of a strategy quantifier. A formula $\phi$ is a \defin{sentence} if $\free(\phi)$ is empty.

\subsection{Semantics}
\label{sec-SLmodels}

 \SLi formulas are evaluated on (possibly infinite) concurrent game arenas
 with  interpretations for observation symbols. 

\begin{defi}
  \label{def-CGSi}
  A \defin{concurrent game arena} (or
  \CGSi) is a tuple
  $\CGSi=(\Act,\vertices,\trans,\val,\pos_\init,\obsint)$ where
   \begin{itemize}
    \item $\Act$ is a finite  set of actions,
    \item $\vertices$ is a   set of positions,
   \item $\trans:\vertices\times \Mov^{\Agf}\to \vertices$ is a transition function, 
  \item $\val:\vertices\to 2^{\APf}$ is a labelling function, 
  \item $\pos_\init \in \vertices$ is an initial position, and
  \item  $\obsint:\Obsf\to 2^{\vertices\times\vertices}$ is an 
  observation interpretation.
  \end{itemize}
\end{defi}

For $\obs\in\Obsf$,  $\obsint(\obs)$
  is an equivalence relation on positions, that we may write
  $\obseq$.  It represents what a player
using a strategy with  observation $\obs$ can see: $\obseq$-equivalent positions
are indistinguishable to a player using a strategy associated with
observation 
$\obs$. 

In a position $\pos\in\setpos$, each player $\ag$ chooses an action $\mova\in\Mov$, 
and the game proceeds to position
$\trans(\pos, \jmov)$, where $\jmov\in \Mov^{\Agf}$ stands for the \defin{joint action}
$(\mova)_{\ag\in\Agf}$. Given a joint action
$\jmov=(\mova)_{\ag\in\Agf}$ and $\ag\in\Agf$, we let
$\jmov(\ag)$ denote $\mova$.
A \defin{play} is an infinite word
$\iplay=\pos_{0} \jmov_0 \pos_{1}\jmov_1 \pos_2\ldots$ such that $\pos_0=\pos_\init$ and
for every $i\geq 0$,
$\trans(\pos_{i}, \jmov_{i})=\pos_{i+1}$. A finite prefix of a play
ending in a position is a \defin{partial play}, {and we let $\FPlay$ be
the set of partial plays.}
For each observation $\obs$ we define the equivalence relation
$\obseq$ on partial plays as follows:  $\pos_{0} \jmov_0 \pos_{1}\jmov_1 \pos_2\ldots\pos_k \obseq \pos'_{0} \jmov'_0 \pos'_{1}\jmov'_1 \pos'_2\ldots\pos'_{k'}$ if $k=k'$,
and $\pos_{i}\obseq\pos'_{i}$ for every $i\in\{0,\ldots,
k\}$.

A  \defin{strategy} is a function {$\strat:\FPlay
\to \Mov$ that maps each partial play to an
action.}  For $\obs\in\Obsf$, an \defin{$\obs$-strategy} is a strategy
$\strat$ such that
 $\strat(\fplay)=\strat(\fplay')$ whenever $\fplay \obseq
\fplay'$. 
We let $\setstrato$ be the set of
all $\obs$-strategies.
An \defin{assignment} is
a partial function $\assign:\Agf\union\Varf \partialto \setstrat$, assigning to
each  player and variable in its domain a strategy.
For an assignment
$\assign$, a player $a$ and a strategy $\strat$,
$\assign[a\mapsto\strat]$ is the assignment of domain
$\dom(\assign)\union\{a\}$ that maps $a$ to $\strat$ and is equal to
$\assign$ on the rest of its domain, and 
$\assign[\var\mapsto \strat]$ is defined similarly, where $\var$ is a
variable. 
In addition, given a formula $\phi\in\SLi$, an assignment is
\emph{variable-complete for $\phi$} if
its domain contains all free variables of $\phi$.

For an assignment $\assign$ and a partial play $\fplay$, we let
$\out(\assign,\fplay)$ be the set of plays that extend
$\fplay$ by letting each player $a$ follow
strategy $\assign(a)$. Formally,
if $\fplay=\pos_0\jmov_0\pos_1\ldots\jmov_{k-1}\pos_k$, then $\out(\assign,\fplay)$ is the set of plays of the form $\fplay \cdot
\jmov_k \pos_{k+1}\jmov_{k+1}\pos_{k+2}\ldots$ such that for all
$i\geq 0$ and all $\ag\in\dom(\assign)\inter\Agf$,
 ${\jmov_{k+i}}(\ag)=\assign(\ag)(\fplay\cdot\jmov_k\pos_{k+1}\ldots\jmov_{k+i-1}\pos_{k+i})$ \mbox{ and }
 $\pos_{k+i+1}=\trans(\pos_{k+i},\jmov_{k+i})$.

\begin{defi}
\label{def-SLi-semantics}
The semantics of a state (resp. path) formula is defined on a \CGSi $\CGSi$, an
assignment  $\assign$ that is variable-complete for $\phi$, and a
partial play $\fplay$ (resp. an infinite play $\iplay'$ and an index $i\in\setn$). The
inductive definition is as follows:
\[
\begin{array}{lcl}
 \CGSi,\assign,\fplay\modelsSL p & \text{ if } & p\in\val(\last(\fplay))\\[1pt]
 \CGSi,\assign,\fplay\modelsSL \neg\phi & \text{ if } &
  \CGSi,\assign,\fplay\not\modelsSL\phi\\[1pt]
 \CGSi,\assign,\fplay\modelsSL \phi\ou\phi' & \text{ if } &
  \CGSi,\assign,\fplay\modelsSL\phi \;\text{ or }\;
  \CGSi,\assign,\fplay\modelsSL\phi' \\[1pt]
 \CGSi,\assign,\fplay\modelsSL\Estrato{\obs}\phi  & \text{ if } & 
\exists \strat\in\setstrato \;\text{s.t.} \;
    \CGSi,\assign[\var\mapsto\strat],\fplay\modelsSL \phi\\[1pt]
 \CGSi,\assign,\fplay\modelsSL \bind{\var}\phi & \text{ if } &
 \CGSi,\assign[\ag\mapsto\assign(\var)],\fplay\modelsSL \phi\\[1pt]  
  \CGSi,\assign,\fplay\modelsSL \Eout\psi & \text{ if } & \exists\iplay' \in
                                                          \out(\assign,\fplay)
                                                          \text{ such
                                                          that }  \CGSi,\assign,\iplay',|\fplay|-1\modelsSL   \psi\\
  \CGSi,\assign,\iplay',i\modelsSL \psi & \text{ if } &
                                                       \CGSi,\assign,\iplay'_{\leq i}\modelsSL\psi\\[1pt]
  \CGSi,\assign,\iplay',i\modelsSL \neg\psi & \text{ if } &
                                                           \CGSi,\assign,\iplay',i\not\modelsSL\psi\\[1pt]
  \CGSi,\assign,\iplay',i\modelsSL \psi\ou\psi' & \text{ if } &
                                                               \CGSi,\assign,\iplay',i\modelsSL\psi \;\text{ or }\;
                                                               \CGSi,\assign,\iplay',i\modelsSL\psi' \\[1pt]
  \CGSi,\assign,\iplay',i\modelsSL\X\psi & \text{ if } &
                                                        \CGSi,\assign,\iplay',i+1\modelsSL\psi\\[1pt]
  \CGSi,\assign,\iplay',i\modelsSL\psi\until\psi' & \text{ if } &
                                                                 \exists j\geq i
                                                                 \mbox{
                                                                 s.t. }\CGSi,\assign,\iplay',j\modelsSL
                                                                 \psi'
                                                                  \text{
                                                                  and } \forall\, k \text{ s.t. } i\leq k <j, \CGSi,\assign,\iplay',k\modelsSL \psi
\end{array}
\]
\end{defi}

A sentence $\phi$ can be evaluated in the empty assignment $\emptyassign$.
Given a sentence $\phi$ and a \CGSi $\CGSi$ with initial position
$\pos_\init$, we write $\CGSi\models\phi$ if $\CGSi,\emptyassign,\pos_\init\models\phi$. 

\subsection{Pushdown game arenas}
\label{sec-pushdown-game-structures}

We introduce
 Pushdown Game Arenas with Visible Stack,
a variant of Epistemic Pushdown Game Structures (\EPGS) defined
in~\cite{chen2017model}, themselves an imperfect-information
generalisation  of the Pushdown Game Structures
from~\cite{murano2015pushdown}. While in \EPGS players have imperfect
information both on the control states and the stack, in
Pushdown Game Arenas with Visible Stack, the stack is perfectly
observed by all players. 
Another minor
difference is that while in \EPGS, observational equivalence relations
are associated to players, in our models they are associated to
observation symbols. 

\begin{defi}
  \label{def-pgs}
A \defin{Pushdown Game Arena with Visible Stack}, or \PGA, is a
tuple $\PGA=(\Act,\stacka,\cstates,\ptrans,\val,\cstate_\init,\obsint)$ where
   \begin{itemize}
    \item $\Act$ is a finite set of actions,
    \item $\stacka$ is a finite stack alphabet together with a  bottom
      symbol $\stackb\notin\stacka$ and we let $\stacka_\stackb=\stacka\cup\{\stackb\}$,
    \item $\cstates$ is a finite set of control states,
    \item $\ptrans:\cstates\times\stacka_\stackb\times \Mov^{\Agf}\to \cstates\times{\stacka_\stackb}^*$ is a transition function, 
    \item $\val:\cstates\times{\stacka}^*\cdot\stackb\to 2^{\APf}$ is a regular labelling function, 
    \item $\cstate_\init \in \cstates$ is an initial control state, and
    \item  $\obsint:\Obsf\to 2^{\cstates\times\cstates}$ is an observation interpretation.
  \end{itemize}  
\end{defi}
As in Definition~\ref{def-pcks}, we  require that the bottom symbol
never be removed
or pushed: for any $\cstate\in\cstates$
 and $\jmov\in\Mov^{\Agf}$, one has
$\ptrans(\cstate,\stackb,\jmov)\in
\cstates\times{\stacka}^*\cdot\stackb$ (the bottom symbol is never
removed), and for every $\stacks\in\stacka$,
$\ptrans(\cstate,\stacks,\jmov)\in \cstates\times{\stacka}^*$ (the
bottom symbol is never pushed).

For $\obs\in\Obsf$, $\obsint(\obs)$ is an equivalence relation on
control states, that we may write $\obseq$.  Also, by \emph{regular}
labelling function, we mean that for each $p\in\APf$ and
$\cstate\in\cstates$,
the set $\{\stackc\in\stacka^*\cdot\stackb\mid
p\in\val(\cstate,\stackc)\}$ forms a regular
language~\cite{DBLP:journals/iandc/EsparzaKS03}.

A \defin{configuration} is a pair
$\conf{\cstate}{\stackc}\in
\cstates\times({\stacka}^*\cdot\stackb)$ where $\cstate$
represents the current control state and $\stackc$ the current content
of the stack.  When the players choose a joint move
$\jmov\in \Mov^{\Agf}$ in a configuration
$\conf{\cstate}{\stacks\cdot\stackc}$ the system moves to
configuration $\conf{\cstate'}{\stackc'\cdot\stackc}$, where
$\conf{\cstate'}{\stackc'}=\ptrans(\cstate,\stacks,\jmov)$; we denote
this by
$\conf{\cstate}{\stacks\cdot\stackc}\conftrans{\jmov}\conf{\cstate'}{\stackc'\cdot\stackc}$.

A \PGA
$\PGA=(\Act,\stacka,\cstates,\ptrans,\val,\cstate_\init,\obsint)$
induces an infinite \CGS
$\CGS_\PGA=(\Act,\vertices',\trans,\val',\pos_\init,\obsint')$ where
\begin{itemize}
\item $\vertices'=\cstates\times ({\stacka}^*\cdot \stackb)$,
\item $\trans(\conf{\cstate}{\stacks\cdot\stackc},\jmov)=
\conf{\cstate'}{\stackc'\cdot\stackc}$ if
$\conf{\cstate}{\stacks\cdot\stackc}\conftrans{\jmov}\conf{\cstate'}{\stackc'\cdot\stackc}$,
\item $\val'=\val$,
\item $\pos'_\init=\conf{\cstate_\init}{\stackb}$,
\item
  $(\conf{\cstate}{\stackc},\conf{\cstate'}{\stackc'})\in\obsint'(\obs)$
  if $\stackc =
\stackc' \mbox{ and }(\cstate,\cstate')\in\obsint(\obs)$.
\end{itemize}

{Plays and partial plays of $\PGA$ are those of $\CGS_\PGA$.}
For an \SLi sentence $\phi$, we write $\PGA\models\phi$ if $\CGS_\PGA\models\phi$.



\subsection{Model checking hierarchical instances}
\label{sec-def-mcproblem}

We study the model-checking problem for \SLi evaluated on pushdown
game arenas with visible stack. This problem is clearly undecidable
as it captures multiplayer games with imperfect information (see for
instance~\cite{DBLP:conf/focs/PetersonR79,PR90}). We generalise a
result from~\cite{DBLP:conf/lics/BerthonMMRV17}, which shows that
model-checking \SLi on finite \CGSs is decidable for
so-called \emph{hierarchical instances}, \ie when each strategy
quantifier in a formula is associated to an observation finer than
those associated to strategy quantifiers higher up in the syntactic tree.

Given an \defin{instance}
$(\PGA,\Phi)$, where $\PGA$ is a \PGA and $\Phi$ is an \SLi sentence, the model-checking problem
 consists in deciding whether $\PGA\models\Phi$.

\begin{defi}
  \label{def-hierarchical-formula}
An instance $(\PGA,\Phi)$ is \defin{hierarchical} if
for every
  subformula $\phi_{1}=\Estrato[y]{\obs_{1}}\phi'_{1}$ of $\Phi$ and subformula
  $\phi_{2}=\Estrato{\obs_{2}}\phi'_{2}$ of $\phi'_1$,
 it holds that 
  $\obsint(\obs_{2})\subseteq \obsint(\obs_{1})$.
\end{defi}



The rest of this section is dedicated to the proof of the following result:

\begin{thm}
\label{theo-SLi}
Model checking \SLi on pushdown game arenas with
visible stack is decidable for hierarchical instances.
\end{thm}

We adapt the reduction
from~\cite{DBLP:conf/lics/BerthonMMRV17} to 
transform hierarchical
instances of \SLi on \PGA into 
 hierarchical instances of \QCTLsi on
\PCKS.
Let $(\PGA,\Phi)$ be a hierarchical instance of the model-checking
problem for \SLi, {and assume without loss of generality that each strategy variable
is quantified at most once in $\Phi$.}

\subsection*{Model transformation} We  first  define the \PCKS
$\PCKS_{\PGA}$. 
Let $\Obsf=\{\obs_{1},\ldots,\obs_{n}\}$, and let
 $\PGA=(\Act,\stacka,\cstates,\ptrans,\val,\cstate_\init,\obsint)$.
 For $i \in [n]$, define
 the local states $\setlstates_{i}\egdef\{\eqc{\obs_{i}}\mid\cstate\in\cstates\}$, where
 $\eqc{\obs}$ is the equivalence class of $\cstate$ for relation
 $\obsint(\obs)$. 
 For each control state $\cstate\in\cstates$ and joint move
 $\jmov\in\Act^\Agf$, we define
 $\sstate_{\cstate,\jmov}\egdef(\eqc{\obs_1},\ldots,\eqc{\obs_n},\cstate,\jmov)$.
Each tuple $\sstate_{\cstate,\jmov}\in
 \prod_{i\in [n]}\setlstates_i\times\cstates\times\Act^\Agf$ contains the equivalence class of $\cstate$
 for each observation $\obs_i\in\Obsf$; we 
 include the exact control state $\cstate$ of $\PGA$ because it is needed to
 define the dynamics, and we also include the last joint action played to 
 make it possible to check that players
 follow their strategies.

Let $\APact=\{p_{\jmov}\mid\jmov\in\Act^\Agf\}$ be a set of
 fresh atomic propositions.
Define the \PCKS
$\PCKS_{\PGA}=(\stacka,\setstates,\relation,\lab',\sstate_{\init})$
over  $\APf\union\APact$, where
\begin{itemize}
\item $\setstates=\{\sstate_{\cstate,\jmov}\mid \cstate\in\cstates
  \mbox{ and }\jmov\in\Act^\Agf\}$,
\item $\relation=\{(\sstate_{\cstate,\jmov},\stacks,\sstate_{\cstate',\jmov'},\stackc')\mid
\ptrans(\cstate,\stacks,\jmov')=(\cstate',\stackc')\}$,
\item
  $\lab'(\conf{\sstate_{\cstate,\jmov}}{\stackc})=\val(\conf{\cstate}{\stackc})\cup
  \{p_{\jmov}\}$, and 
\item $\sstate_{\init}=\sstate_{\cstate_\init,\jmov_\init}$ for some arbitrary
  $\jmov_\init\in\Act^\Agf$.
\end{itemize}

The labelling $\lab'$ is regular
 because $\lab$ is
regular for atoms in $\APf$, and the truth value of atoms
in $\APact$ is determined by the control state only. 

For every partial play
$\fplay=\conf{\cstate_\init}{\stackb}\jmov_0\conf{\cstate_1}{\stackc_1}\ldots\conf{\cstate_k}{\stackc_k}$
in $\PGA$, define the partial path
$\spath'=\conf{\sstate_{\init}}{\stackb}\conf{\sstate_1}{\stackc_1}\ldots\conf{\sstate_k}{\stackc_k}$ in $\PCKS_\PGA$
  where $\sstate_i=\sstate_{\cstate_i,\jmov_{i-1}}$, for each $i\in [k]$.
The mapping $\fplay\mapsto\spath'$ puts in bijection 
 partial plays of $\CGS_\PGA$ with partial paths
of $\PCKS_\PGA$.

\subsection*{Formula transformation}
 We now describe how to transform an \SLi formula $\phi$ and a partial
function $f:\Agf \partialto  \Varf$ into a \QCTLsi
formula $\tr[f]{\phi}$ (that will also depend on $\PGA$).
Suppose that $\Mov=\{\mov_{1},\ldots,\mov_{\maxmov}\}$, and define
$\tr[f]{\phi}$ and $\trp[f]{\psi}$ by mutual induction on state and path formulas. 

Base, boolean and temporal cases are as follows:
\begin{align*}
\tr[f]{p} 		 & \egdef p & \trp[f]{\phi} & \egdef
\tr[f]{\phi}\\
\tr[f]{\neg \phi} 	 & \egdef \neg \tr[f]{\phi} & \trp[f]{\neg
  \psi} & \egdef \neg \trp[f]{\psi}\\
\tr[f]{\phi_1\ou\phi_2}  &\egdef \tr[f]{\phi_1}\ou\tr[f]{\phi_2} &
\trp[f]{\psi_1\ou\psi_2} &\egdef \trp[f]{\psi_1}\ou\trp[f]{\psi_2} &\\
\trp[f]{\X\psi}  & \egdef \X\trp[f]{\psi} &
                                            \trp[f]{\psi_1\until\psi_2}  & \egdef \trp[f]{\psi_1}\until\trp[f]{\psi_2}.
\end{align*}

For the strategy quantifier we let
\[\tr[f]{\Estrato{\obs}\phi} \egdef  \exists^{\trobs{\obs}} p_{\mov_{1}}^{\var}\ldots \exists^{\trobs{\obs}} p_{\mov_{\maxmov}}^{\var}. \phistrat \et \tr[f]{\phi},\]
where $\trobs{\obs_i} \egdef \{j\mid \obsint(\obs_{i})\subseteq\obsint(\obs_{j})\}$ and
$\phistrat$, which checks that atoms $p_{\mov}^\var$ indeed code for a
strategy, is defined as
\[\phistrat \egdef
\A\always\bigou_{\mov\in\Mov}(p_{\mov}^{\var}\et\biget_{\mov'\neq\mov}\neg
p_{\mov'}^{\var}).\]

Let $\tr[f]{\bind{\var}\phi} \egdef \tr[{f[\ag\mapsto \var]}]{\phi}$,
and
$\tr[f]{\Eout\psi}\egdef \E\,(\psiout[f] \wedge
 \trp[f]{\psi})$, 
where
\[\psiout[f]\egdef \always
 \bigou_{\jmov\in\Mov^{\Agf}} 
  \biget_{\ag\in\dom(f)}p_{\jmov(\ag)}^{f(\ag)}\et \X\,
  p_{\jmov}\]

Formula $\psiout[f]$ holds  on a path
if and only if each player $\ag$ in $\dom(f)$
follows the  strategy coded by atoms $p^{f(\ag)}_\act$.

The correctness of the translation is stated by  the following lemma:
 \begin{lem}
   \label{lem-redux}
     \[\PGA \models \Phi \mbox{ if, and only if, }
\PCKS_\PGA \models
 \tr[\emptyset]{\Phi}.\]
 \end{lem}

 To establish this lemma we need a few additional definitions. Given a strategy $\strat$ and
a strategy variable $\var$ we
let  $\stratlab{\var}\egdef\{\plab[{p_\act^\var}]\mid
\act\in\Act\}$ be the family of $p_\act^\var$-\labelings for tree
$\unfold[\PCKS_{\PGA}]{}$ defined as follows: for each
finite play $\fplay$ in $\PGA$ and $\act\in\Act$,
we let $\plab[{p_\act^\var}](\noeud_\fplay)\egdef 1$ if $\act=\strat(\fplay)$, 0 otherwise.
For a \labeled tree $\ltree$ with same domain as
$\unfold[\PCKS_{\PGA}]{}$ we write $\ltree\prodlab \stratlab{\var}$ for
$\ltree\prodlab \plab[{p_{\act_1}^\var}]\prodlab\ldots\prodlab \plab[{p_{\act_\maxmov}^\var}]$.

{Given a partial play $\fplay$ in $\PGA$, we define the node
$\noeud_\fplay=\succpath(\fplay')
\in\unfold[\PCKS_{\PGA}]{\sstate_{\pos_\init}}$: it is the
 succinct representation of $\fplay'$, the finite path of $\PCKS_\PGA$
 that corresponds to $\fplay$. Also,}
given an infinite play $\iplay$ and a point
$i\in\setn$, we let
$\tpath_{\iplay,i}$ be the infinite path in
$\unfold[\PCKS_{\PGA}]{\sstate_{\pos_\init}}$ that starts in node
$\noeud_{\iplay_{\leq
    i}}$ and is defined as
$\tpath_{\iplay,i}\egdef\noeud_{\iplay_{\leq i}}\noeud_{\iplay_{\leq
    i+1}}\noeud_{\iplay_{\leq i+2}}\ldots$

{Finally, for an assignment $\assign$ and
a partial   function $f:\Agf\partialto\Varf$, we say that $f$ is
\emph{compatible} with $\assign$ if
 $\dom(f)=\dom(\assign)\inter \Agf$ and  for all $a \in \dom(f)$,  $\assign(a) = \assign(f(a))$.}

 Lemma~\ref{lem-equivalence} is now obtained by applying the following result
 to sentence $\Phi$, $\fplay=\pos_\init$, the empty
 assignment and
 the empty function $\emptyset$:

 \begin{prop}
   \label{prop-redux}
   For every  state subformula $\phi$ and path subformula $\psi$ of
   $\Phi$, partial play $\fplay$,  play $\iplay'$, point
   $i\in\setn$, for every  assignment $\assign$ variable-complete
   for $\phi$ (resp. $\psi$) and
partial   function $f:\Agf\partialto\Varf$ compatible with $\assign$, assuming
   also that no $\var_i$ in $\dom(\assign)\inter \Varf=\{\var_1,\ldots,\var_k\}$ is
 quantified in $\phi$ or $\psi$, we have
\[  
\PGA,\assign,{\fplay}\models\phi  \quad\mbox{ iff } \quad
  \unfold[\PCKS_{\PGA}]{\sstate_{\pos_\init}}\prodlab
  \stratlab[\assign(\var_1)]{\var_1}\prodlab\ldots \prodlab
  \stratlab[\assign(\var_k)]{\var_k},\noeud_{\fplay} \modelst
       \tr[f]{\phi}
    \]
   and
\[ \PGA,\assign,{\iplay'},i\models\psi \quad\mbox{ iff } \quad
  \unfold[\PCKS_{\PGA}]{\sstate_{\pos_\init}}\prodlab
  \stratlab[\assign(\var_1)]{\var_1}\prodlab\ldots \prodlab
  \stratlab[\assign(\var_k)]{\var_k},\tpath_{\iplay',i} \modelst
  \trp[f]{\psi}     
  \]
 \end{prop}

 \begin{proof}
   The proof is by induction on $\phi$.
We detail the cases for binding,  strategy quantification and
outcome quantification, the others follow simply by definition of
$\PCKS_{\PGA}$ for atomic propositions and induction hypothesis for
remaining cases.

{For $\phi=\bind{\var}\phi'$, we have
$\PGA,\assign,{\fplay}\models\bind{\var}\phi'$ iff 
$\PGA,\assign[\ag\mapsto\assign(\var)],{\fplay}\models\phi'$.
The result follows by using the induction hypothesis with assignment
$\assign[\ag\mapsto\var]$ and function
 $f[a\mapsto\var]$. This is possible because $f[a\mapsto\var]$ is compatible with $\assign[\ag\mapsto\var]$: indeed
 $\dom(\assign[\ag\mapsto\var])\inter\Agf$ is equal to
 $(\dom(\assign)\inter\Agf) \union \{a\}$ which, by assumption, is equal
 to $\dom(f) \union \{a\}=\dom(f[a\mapsto
 \var])$. Also
 by assumption, for all $a'\in\dom(f)$, $\assign(a')=\assign(f(a'))$, and 
by definition \[\assign[a\mapsto \assign(\var)](a)=\assign(\var)=\assign(f[a\mapsto\var](a)).\]}

 {For $\phi=\Estrato{\obs}\phi'$, assume first that
$\PGA,\assign,{\fplay}\models\Estrato{\obs}\phi'$. There exists an
$\obs$-uniform strategy $\strat$ such that
\[\PGA,\assign[\var\mapsto \strat],\fplay\models \phi'.\] Since $f$ is
compatible with $\assign$, it is also compatible with assignment
$\assign'=\assign[\var\mapsto \strat]$. By assumption, no variable in
$\{\var_1,\ldots,\var_k\}$ is quantified in $\phi$, so that $\var\neq
\var_i$ for all $i$ and thus $\assign'(\var_i)=\assign(\var_i)$ for
all $i$; and because no strategy variable is
quantified twice in a same formula,
$\var$ is not quantified in $\phi'$, so that no variable in
$\{\var_1,\ldots,\var_k,\var\}$ is quantified in $\phi'$.
By induction hypothesis 
  \[\unfold[\PCKS_{\PGA}]{\sstate_{\pos_\init}}\prodlab
  \stratlab[\assign'(\var_1)]{\var_1}\prodlab\ldots \prodlab
  \stratlab[\assign'(\var_k)]{\var_k}\prodlab   \stratlab[\assign'(\var)]{\var},\noeud_{\fplay}
  \modelst \tr[f]{\phi'}.\]

  Because $\strat$ is $\obs$-uniform, each
  $\plab[{p_\act^\var}]\in\stratlab{\var}=\stratlab[\assign(\var)]{\var}$ is $\trobs{\obs}$-uniform,
  and it follows that   \[\unfold[\PCKS_{\PGA}]{\sstate_{\pos_\init}}\prodlab
  \stratlab[\assign'(\var_1)]{\var_1}\prodlab\ldots \prodlab
  \stratlab[\assign'(\var_k)]{\var_k},\noeud_{\fplay}
  \modelst \exists^{\trobs{\obs}} p_{\mov_{1}}^{\var}\ldots
  \exists^{\trobs{\obs}} p_{\mov_{\maxmov}}^{\var}. \phistrat\et\tr[f]{\phi'}.\]

Finally, since $\assign'(\var_i)=\assign(\var_i)$ for all $i$, we
conclude that 
\[\unfold[\PCKS_{\PGA}]{\sstate_{\pos_\init}}\prodlab
  \stratlab[\assign(\var_1)]{\var_1}\prodlab\ldots \prodlab
  \stratlab[\assign(\var_k)]{\var_k},\noeud_{\fplay}
  \modelst \tr[f]{\Estrato{\obs}\phi'}.\]

For the other direction, assume
that
\[\unfold[\PCKS_{\PGA}]{\sstate_{\pos_\init}}\prodlab
  \stratlab[\assign(\var_1)]{\var_1}\prodlab\ldots \prodlab
  \stratlab[\assign(\var_k)]{\var_k},\noeud_{\fplay} \modelst
  \tr[f]{\phi},\] and recall that
$\tr[f]{\phi}=\exists^{\trobs{\obs}} p_{\mov_{1}}^{\var}\ldots
\exists^{\trobs{\obs}}
p_{\mov_{\maxmov}}^{\var}. \phistrat\et\tr[f]{\phi'}$.  Write
$\ltree=\unfold[\PCKS_{\PGA}]{\sstate_{\pos_\init}}\prodlab
\stratlab[\assign(\var_1)]{\var_1}\prodlab\ldots \prodlab
\stratlab[\assign(\var_k)]{\var_k}$. There exist
$\trobs{\obs}$-uniform $\plab[{p_\act^\var}]$-\labelings such that
\[\ltree\prodlab  \plab[{p_{\act_1}^\var}]\prodlab\ldots\prodlab \plab[{p_{\act_\maxmov}^\var}]
  \modelst \phistrat\et\tr[f]{\phi'}.\]
By $\phistrat$, these \labelings  
 code for a strategy $\strat$, and because they are
 $\trobs{\obs}$-uniform, $\strat$ is $\obs$-uniform. Let
 $\assign'=\assign[\var\mapsto \strat]$. For all $1\leq i\leq k$, by
 assumption $\var\neq \var_i$, and thus $\assign'(\var_i)=\assign(\var_i)$.
 The above can thus be rewritten
 \[\unfold[\PCKS_{\PGA}]{\sstate_{\pos_\init}}\prodlab
\stratlab[\assign'(\var_1)]{\var_1}\prodlab\ldots \prodlab
\stratlab[\assign'(\var_k)]{\var_k}\prodlab  \stratlab[\assign'(\var)]{\var}
  \modelst \phistrat\et\tr[f]{\phi'}.\]
 By induction hypothesis we have
$\PGA,\assign[\var\mapsto
\strat],\fplay\models \phi'$, hence $\PGA,\assign,\fplay\models \Estrato{\obs}\phi'$.}

For $\phi=\Eout\psi$,
assume first that $\PGA,\assign,{\fplay}\models\E\psi$. 
There exists an infinite play $\iplay'\in\out(\assign,\fplay)$ s.t.
$\PGA,\assign,\iplay',|\fplay|-1\modelsSL \psi$. By induction
hypothesis,
\[\unfold[\PCKS_{\PGA}]{\sstate_{\pos_\init}}\prodlab
  \stratlab[\assign(\var_1)]{\var_1}\prodlab\ldots \prodlab
  \stratlab[\assign(\var_k)]{\var_k},\tpath_{\iplay',|\fplay|-1} \modelst
  \trp[f]{\psi}.\]
Since $\iplay'$ is an outcome of $\assign$, each agent $a\in\dom(\assign)\inter\Agf$ 
follows strategy $\assign(a)$ in $\iplay'$.
Because  $\dom(\assign)\inter \Agf=\dom(f)$ and for all $a \in \dom(f)$,
  $\assign(a) = \assign(f(a))$, each agent $a\in\dom(f)$ follows
the  strategy $\assign(f(a))$, which is coded by atoms
$p_\mov^{f(\ag)}$ in the translation of $\Phi$. Therefore $\tpath_{\iplay',|\fplay|-1}$ also
satisfies $\psiout$, hence $\unfold[\PCKS_{\PGA}]{\sstate_{\pos_\init}}\prodlab
  \stratlab[\assign(\var_1)]{\var_1}\prodlab\ldots \prodlab
  \stratlab[\assign(\var_k)]{\var_k},\tpath_{\iplay',|\fplay|-1} \modelst
  \psiout \et   \trp[f]{\psi}$, and we are done.

  For the other direction, assume that 
  \[\unfold[\PCKS_{\PGA}]{\sstate_{\pos_\init}}\prodlab
  \stratlab[\assign(\var_1)]{\var_1}\prodlab\ldots \prodlab
  \stratlab[\assign(\var_k)]{\var_k},\noeud_\fplay \modelst
  \E(\psiout[f] \et   \trp[f]{\psi}).\]
There exists a path $\tpath$ in $\unfold[\PCKS_{\PGA}]{\sstate_{\pos_\init}}\prodlab
  \stratlab[\assign(\var_1)]{\var_1}\prodlab\ldots \prodlab
  \stratlab[\assign(\var_k)]{\var_k}$ starting in
node $\noeud_\fplay$ that satisfies both $\psiout[f]$ and $\trp[f]{\psi}$.
By construction of $\PCKS_{\PGA}$ and definition of succinct unfoldings, there exists an infinite play $\iplay'$
such that $\iplay'_{\leq |\fplay|-1}=\fplay$ and $\tpath=\tpath_{\iplay',|\fplay|-1}$.
By induction hypothesis, $\PGA,\assign,\iplay',|\fplay|-1 \modelsSL \psi$.
Because $\tpath_{\iplay',|\fplay|-1}$ satisfies $\psiout[f]$, $\dom(\assign)\inter \Agf=\dom(f)$, and for all $a \in \dom(f)$,
  $\assign(a) = \assign(f(a))$, it is also the case that
  $\iplay'\in\out(\assign,\fplay)$, 
hence  $\PGA,\assign,\fplay \modelsSL \Eout\psi$.
 \end{proof}

 To complete the proof of Theorem~\ref{theo-SLi} it remains to check that $\tr[\emptyset]{\Phi}$ is a
hierarchical \QCTLsi formula, which is the case because
 $\Phi$ is hierarchical in
$\PGA$ and for every two observations $\obs_{i}$ and $\obs_{j}$ in $\Obsf$ such that
$\obsint(\obs_{i})\subseteq\obsint(\obs_{j})$, by definition of $\trobs{\obs_{k}}$
we have that $\trobs{\obs_{i}}\subseteq \trobs{\obs_{j}}$.

\section{Higher-order extension}\label{sec-ho}

We have shown so far that the techniques
developed for \emph{finite} concurrent game arenas with imperfect
information in \cite{DBLP:conf/lics/BerthonMMRV17} can be extended and
adapted to deal with the case of infinite concurrent game arenas
defined by \emph{pushdown} {systems} when the stack is visible. In particular
we proved in Theorem~\ref{theo-SLi} that the model-checking problem
for \SLi on pushdown game arenas with visible stack is decidable for
hierarchical instances. Moving from finite structures to infinite
structures (in our case defined by pushdown {systems}) is of interest
for instance when dealing with system verification as it permits to
capture richer classes, in particular those coming from programs
making use of recursion.

A natural line of research is to go beyond pushdown
{systems}, and a natural candidate here is to move to the
  higher-order setting, \ie to consider higher-order pushdown {systems}
  or even collapsible pushdown {systems}~\cite{HagueMOS17}. These are
  very natural models in particular regarding application for programs
  using higher-order
  functions.

We first briefly discuss the global road map. 
\begin{itemize}
\item The decidability proof for \SLi will again go through a reduction to model checking hierarchical \QCTLsi and this is where most technicalities are coming. 
\item We give in sections~\ref{sec:HO-stack} and
  \ref{sec:HO-stack-links-app} definition of higher-order stacks and
  stacks with links. 
\item Next in Section~\ref{sec-CKS-HO}, we adapt Definition~\ref{def-pcks} and introduce higher-order pushdown compound Kripke structures and collapsible pushdown compound Kripke structures. The main technicality here is to introduce a suitable notion of regular labelling functions.
\item In Section~\ref{sec:CPTA} we explain (in Section~\ref{sec:CPTA-def}) how to generalise to higher-order the definitions of pushdown tree automata from Section~\ref{sec-direction-guided} and we also adapt the results concerning projection, simulation and narrowing (in Section~\ref{sec:CPTA-res}).
\item In Section~\ref{sec-QCTL-HO}, we adapt the notion of succinct unfolding to handle higher-order. This, together with a closure property for alternating collapsible pushdown tree automata, permits to establish decidability of model checking hierarchical \QCTLsi on collapsible pushdown compound Kripke structures.
\item Finally, in Section~\ref{sec-SLi-HO}, we prove that the
  model-checking problem for \SLi on collapsible pushdown game arenas
  with visible stack is decidable for hierarchical instances.
\end{itemize}

\subsection{Higher-order and collapsible pushdown compound Kripke structures}

We explain how to adapt the definitions from Section~\ref{sec-PCKS} to deal with higher-order stacks (possibly with links).


\subsubsection{Higher-order stacks and their operations}\label{sec:HO-stack}

Fix a finite stack alphabet $\Gamma$ and a distinguished \emph{bottom
  symbol} $\bot \not\in \Gamma$.
\begin{defi}
 \label{def-stacks}
An order-1 stack is a word
$\bot a_1 \ldots a_\ell \in \bot \cdot \Gamma^*$
which is denoted $\mksk{\bot a_1 \ldots a_\ell}_1$. An \defin{order-$k$ stack}
(or a \defin{$k$-stack}), for $k>1$, is a non-empty sequence $\stackc_1,\ldots,\stackc_\ell$ of
order-$(k\!-\!1)$ stacks which is written $\mksk{\stackc_1 \ldots \stackc_\ell}_k$.  
\end{defi}
 For
convenience, we may sometimes see an element $a\in\Gamma$ as an order-$0$ stack,
denoted $\mksk{a}_0$.
{We define $\bot_h$, the
\defin{empty $h$-stack}, as: $\bot_0 = \bot$ and
$\bot_{h+1} = \mksk{\bot_h}$.}
We denote by $\Stacks_k$ the set of all order-$k$ stacks and $\Stacks=\bigcup_{k
\geq 1} \Stacks_k$ the set of all higher-order stacks. The height of the stack $\stackc$, denoted
$\len{\stackc}$, is simply the 
length of the sequence. We denote by $\order(\stackc)$ the order of the stack $\stackc$.


In addition to the operations $\pushone{a}$ and $\popn{1}$ that respectively pushes and 
 pops a symbol in the topmost order-$1$ stack, one needs extra operations to deal with the higher-order stacks: the $\popn{k}$ operation removes the topmost order-$k$ stack, while the $\pushn{k}$ duplicates it.

For an order-$n$ stack $\stackc = \mksk{\stackc_1 \ldots \stackc_\ell}_n$ 
and an order-$k$ stack $\stackca$ with $0\leq k<n$, we define $\stackc \pp \stackca$ 
as the order-$n$ stack obtained by pushing $\stackca$ on top of $\stackc$:
\[
\stackc \pp \stackca = 
\left\{
 \begin{array}{lcl}
 \mksk{\stackc_1 \ldots \stackc_\ell \,\stackca}_n & & \textrm{if $k=n-1$,} \\
 \mksk{\stackc_1 \ldots (\stackc_\ell \pp \stackca)}_n & & \textrm{otherwise.} \\
 \end{array}
\right.
\]

We first define the (partial) operations $\popn{i}$ and $\topn{i}$ with $i \geq
1$: $\topn{i}(\stackc)$ returns the top $(i-1)$-stack of $\stackc$, and
$\popn{i}(\stackc)$ returns $\stackc$ with its top $(i-1)$-stack
removed. Formally, for an order-$n$ stack $\mksk{\stackc_1 \cdots \stackc_{\ell+1}}_n$ with
$\ell\geq 0$,
\[\begin{array}{rll}
\topn{i}(\stackc
) & = &
\left\{\begin{array}{ll} \stackc_{\ell+1} & \hbox{if $i = n$}\\
 \topn{i}(\stackc_{\ell+1}) \quad & \hbox{if $i < n$}
\end{array}\right.\\[15pt]
\popn{i}(\stackc
) & = & 
\left\{\begin{array}{ll}
\mksk{\stackc_1 \cdots \stackc_\ell}_n & \hbox{if $i = {n}$ and $\ell \geq 1$}\\
\mksk{\stackc_1 \cdots \stackc_\ell \,
\popn{i}(\stackc_{\ell+1})}  & \hbox{if $i < {n}$}
\end{array}\right.\\
\end{array}\]

By abuse of notation, we let $\topn{\order(\stackc)+1}(\stackc) = \stackc$. Note that
$\popn{i}(\stackc)$ is defined if and only if the height of $\topn{i+1}(\stackc)$ is
strictly greater than $1$. For example $\popn{2}(\mksk{\mksk{\bot \, a \,
b}_1}_2)$ is
undefined.

We now introduce the operations $\pushn{i}$ with $i\geq 2$ that duplicates the
top $(i-1)$-stack of a given stack. More precisely, for an order-$n$ stack $\stackc$
and for $2\leq i\leq n$, we let $\pushn{i}(s)=\stackc\pp \topn{i}(\stackc)$.

The last operation, $\pushone{a}$ pushes the symbol $a\in\Gamma$ on top of the
top $1$-stack. 
More precisely, for an order-$n$ stack $\stackc$ and for a symbol $a\in\Gamma$, we let
$\pushone{a}(\stackc)=\stackc\pp \mksk{a}_0$.

\begin{exa}
  Let $\stackc$ be the following $3$-stack of height $2$:
  \[\mksk{
 \mksk{
 \mksk{\bot b a a c}_{1}
 \mksk{\bot b c c}_{1}
 \mksk{\bot c b a}_{1}
 }_2
 \mksk{
 \mksk{\bot b a a}_{1}
 \mksk{\bot b a b}_{1}
 }_2
}_3\]
Then $\topn{3}(\stackc)$ is the $2$-stack \[\mksk{
 \mksk{\bot b a a}_{1}
 \mksk{\bot b a b}_{1}
 }_2\] 
and 
$\popn{3}(\stackc)$ is~the stack \[\stackc'=\mksk{
 \mksk{
 \mksk{\bot b a a c}_{1}
 \mksk{\bot b b}_{1}
 \mksk{\bot c b a}_{1}
 }_2
}_3\] Note that $\popn{3}(\popn{3}(\stackc))$ is undefined. Then
$\pushn{2}(\stackc')$ is the stack \[\mksk{
 \mksk{
 \mksk{\bot b a a c}_{1}
 \mksk{\bot b b}_{1}
 \mksk{\bot c b a}_{1}
 \mksk{\bot c b a}_{1}
 }_2
}_3\]
and 
\[\pushone{c}(\stackc') = \mksk{
 \mksk{
 \mksk{\bot b a a c}_{1}
 \mksk{\bot b b}_{1}
 \mksk{\bot c b a c}_{1}
 }_2
}_3\]
\end{exa}

\subsubsection{Stacks with links and their operations}\label{sec:HO-stack-links-app}

We now define a richer structure of higher-order stacks where we allow links. Intuitively, a stack with links is a higher-order stack in which any symbol
may have a link that points to an internal stack below it. This link may be used
later to collapse part of the stack.

Order-$k$ stacks with links are order-$k$ stacks with a richer stack
alphabet. Indeed, each symbol in the stack can be either an element $a\in\Gamma$
(\ie it is not the source of a link) or an element
$(a,\ell,h)\in\Gamma\times\{2,\cdots,k\}\times\nat$ (\ie it is
the source of an $\ell$-link pointing to the $h$-th $(\ell-1)$-stack inside the
topmost $\ell$-stack {below the source of the link}). 
Formally, order-$k$ stacks with links over alphabet $\Gamma$ are defined as
order-$k$ stacks
\footnote{Note that we therefore slightly generalise our previous definition as
we implicitly use an infinite stack alphabet, but this does not introduce any
technical change in the definition.}
 over alphabet $\Gamma\cup\Gamma\times\{2,\cdots,k\}\times\nat$.

\begin{exa}\label{ex:HOStackWLinks}
Stack $\stackc$ below is an order-$3$ stack with links:
\[\mksk{
 \mksk{
 \mksk{\bot b}_{1}
 \mksk{\bot b c (c,2,1)}_{1}
 }_2
 \mksk{
 \mksk{\bot a}_{1}
 \mksk{\bot b c}_{1}\\
 \mksk{\bot b (a,2,1) (b,3,1)}_{1}
 }_2
}_3.\]

To improve readability when displaying $n$-stacks in examples, we shall
explicitly draw the links rather than using stack symbols in
$\Gamma\times\{2,\cdots,k\}\times\nat$. For instance, we  represent
$\stackc$ as follows:

\begin{center}\includegraphics[scale=.4]{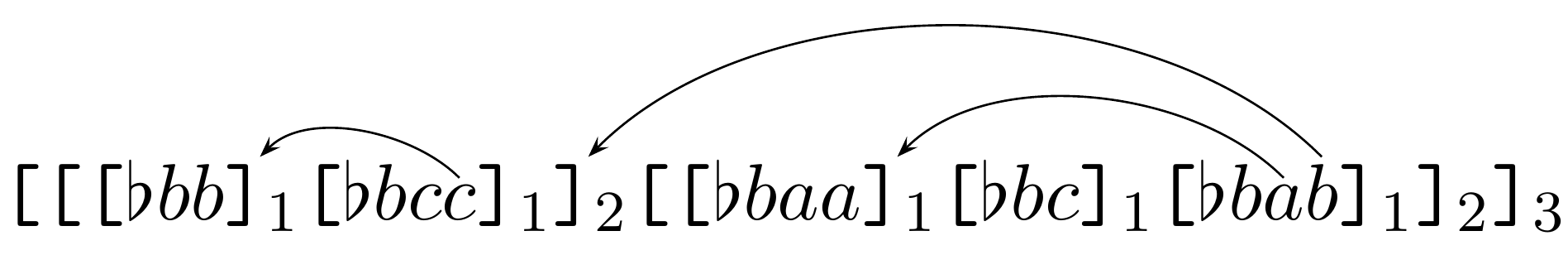}\end{center}

\end{exa}

In addition to the previous operations $\popn{i}$, $\pushn{i}$ and
$\pushone{a}$, we introduce two extra operations: one to create links, and the
other to collapse the stack by following a link.
Link creation is made when pushing a new stack symbol, and the target of an
$\ell$-link is always the $(\ell-1)$-stack below the topmost one. Note that due to possible subsequent copies links can point to arbitrarily deep stacks. 
Formally, we define
$\pushlk{\ell}{a}(\stackc) = \pushone{(a,\ell,h)}$ where we let
$h=|\topn{\ell}(\stackc)|-1$ and require that $h>1$. 

The collapse operation is defined only when the topmost symbol is the source of
an $\ell$-link, and results in truncating the topmost $\ell$-stack to only keep
the component below the target of the link.
Formally, if $\topn{1}(\stackc)=(a,\ell,h)$ and $\stackc=\stackca\pp[\stackcb_1\cdots \stackcb_k]_{\ell}$ with
$k>h$ we let $\collapse(\stackc)=\stackca\pp[\stackcb_1\cdots \stackcb_h]_{\ell}$.

For any $k$, we let $\opn{k}(\Gamma)$ denote the set of all operations over order-$k$ stacks with links.

\begin{exa}\label{eg:3stack}
Let $\stackc = \mksk{\mksk{\mksk{ \, \bot \, a}_1}_2 \;
  \mksk{\mksk{ \, \bot}_1 \mksk{ \, \bot \, a}_1}_2}_3$. We have
  
  \begin{center}\includegraphics[scale=.5]{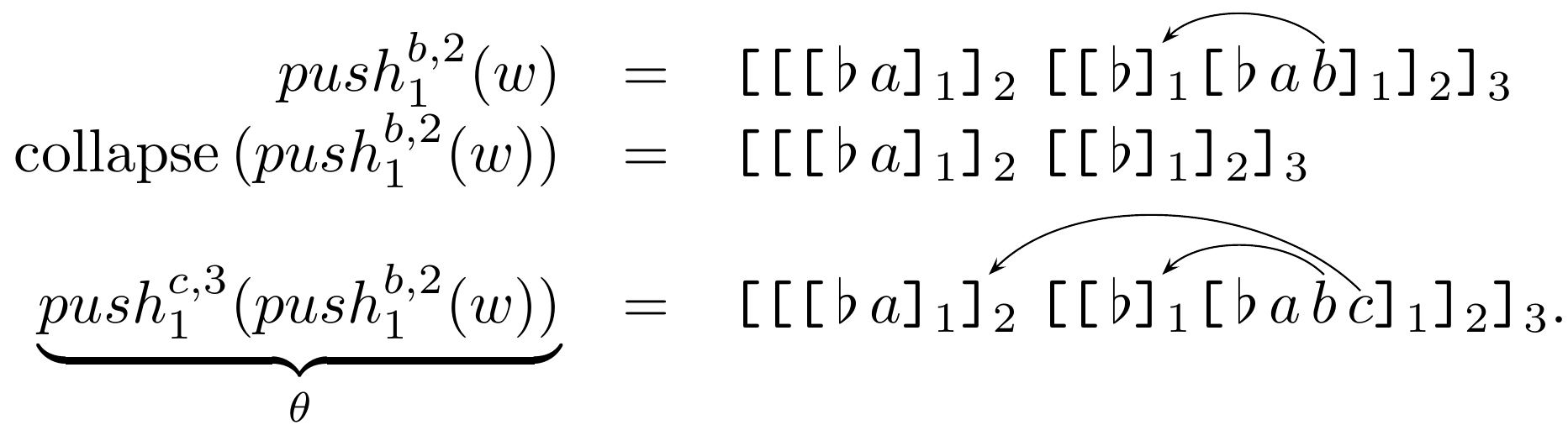}\end{center}
%
 Then $\pushn{2}(\theta)$ and $\pushn{3}(\theta)$ are respectively
  \begin{center}\includegraphics[scale=.4]{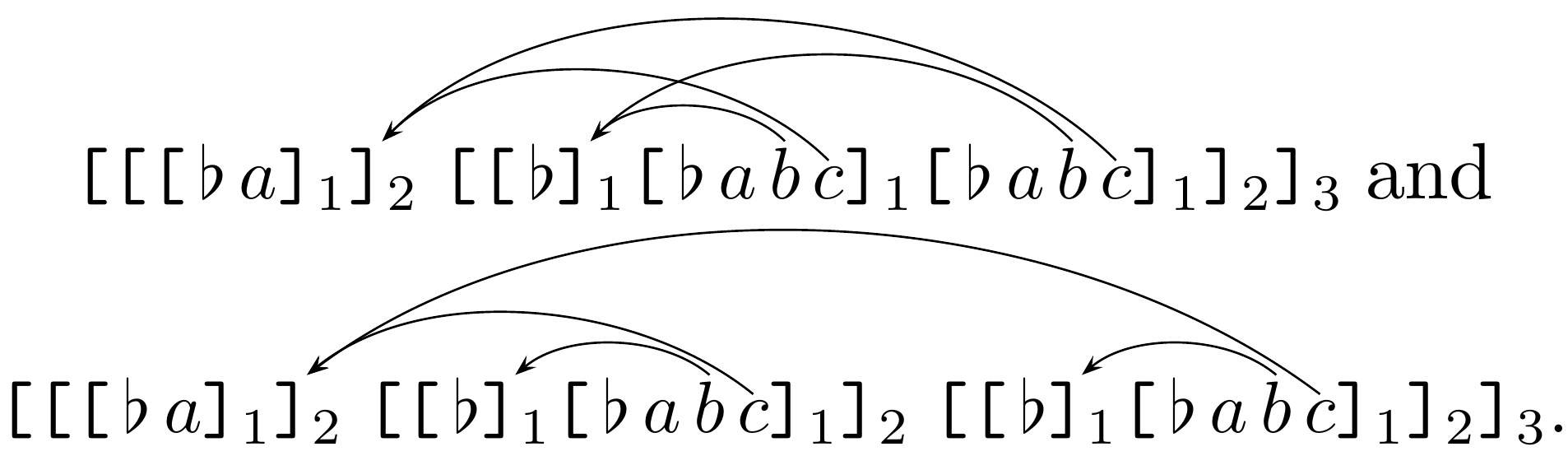}\end{center}
%

\noindent We have
\begin{align*}
\collapse \, (\pushn{2} (\theta)) &= \collapse \, (\pushn{3}(
\theta)) \\ & = \collapse(\theta) \\
 & = \mksk{\mksk{\mksk{ \, \bot \, a}_1}_2}_3.  
\end{align*}
 \end{exa}

\subsubsection{Higher-order and collapsible pushdown compound Kripke structures}\label{sec-CKS-HO}

We are now ready to generalise Definition~\ref{def-pcks} (pushdown
compound Kripke structures) to higher-order. Note that pushdown
compound Kripke structures will correspond to order-1 collapsible
pushdown compound Kripke structures.

\begin{defi}
  \label{def-hopcks}
An order-$k$ \emph{collapsible pushdown compound Kripke structure} or \CPCKS, over
local states $\{\setlstates_i\}_{i\in [n]}$, is a tuple 
$\CPCKS=(\stacka,\setstates,\relation,\lab,\sstate_\init)$ where
\begin{itemize}
    \item $\stacka$ is a finite stack alphabet together with a  bottom
      symbol $\stackb\notin\stacka$, and we let $\stacka_\stackb=\stacka\cup\{\stackb\}$;
\item $\setstates\subseteq \prod_{i\in [n]}\setlstates_i$  is a finite
  set of
states; 
\item $\relation\subseteq\setstates\times\stacka_\stackb\times\setstates\times\opn{{k}}(\Gamma)$ is a
 transition
relation; 
\item $\lab:\setstates\times\Stacks_k \to 2^{\APf}$ is a regular labelling function (defined below);
\item $\sstate_\init \in \setstates$ is an initial state.
\end{itemize}
A \emph{higher-order pushdown compound Kripke structure} (\HOPCKS) is a collapsible pushdown compound Kripke structure that never uses the collapse operation.
\end{defi}

A \defin{configuration} is a pair
$\config=\conf{\sstate}{\stackc}\in\setstates\times\Stacks_k $ where
$\sstate$ is the current state and $\stackc\in\Stacks_k$ the current
stack with links; we call $\conf{\sstate_\init}{\bot_k}$ the
\defin{initial configuration}.

{From configuration
  $\conf{\sstate}{\stackc}$ with $\topone(\stackc)=\gamma$ the system
  can move to $\conf{\sstate'}{op(\stackc)}$ if
  $(\sstate,\gamma,\sstate',op)\in\relation$, which we write
  $\conf{\sstate}{\stackc}\conftrans{}\conf{\sstate'}{op(\stackc)}$. We
  assume that for every configuration $\conf{\sstate}{\stackc}$ there
  exists at least one configuration $\conf{\sstate'}{\stackc'}$ such
  that $\conf{\sstate}{\stackc}\conftrans{}\conf{\sstate'}{\stackc'}$.

  A \defin{path} in $\CPCKS$ is an infinite sequence of configurations
  $\spath=\config_{0}\config_{1}\ldots$ such that $\config_0$ is the
  initial configuration and for all $i\in\setn$,
  $\config_i\conftrans{}\config_{i+1}$.  A \defin{partial path} is a
  finite non-empty prefix of a path.  We let $\Paths(\CPCKS)$
  (resp. $\PPaths(\CPCKS)$) be the set of all paths (resp. partial
  paths) in $\CPCKS$.}

\paragraph{Regular labelling functions}\label{sec-regular-labelling-HO-app}

We need to adapt the concept of regular labelling functions to the
higher-order setting. In the case of pushdown compound Kripke
structures, recall that the criterion used was whether the stack
content belongs to a regular language. Equivalently, one could have
used an MSO-logic formula or a $\mu$-calculus formula on words (as
these frameworks are equivalent to finite-state automata when defining
sets of words). In the higher-order case, we take a similar approach,
\ie we consider a  model of automata working on higher-order stacks 
(\emph{resp.} stacks with links) that is equivalent with the $\mu$-calculus when defining sets of
higher-order stacks (\emph{resp.} stack with links). Note that it is
not equivalent with MSO-logic, which is in fact undecidable over
collapsible pushdown Kripke structures. We start by first giving the
definition for higher-order pushdown compound Kripke structures and
then move to collapsible pushdown compound Kripke structures.

In the (simpler) case of higher-order pushdown compound Kripke
structures, a \defin{regular labelling function} is given as a set of
finite word automata $\wautop{\sstate}$ over alphabet
$\stacka\cup\{[,]\}$, one for each proposition $p\in\APf$ and each
state $\sstate\in\setstates$. They define the labelling function that maps
to each state $\sstate\in\setstates$ and higher-order stack content
$\stackc\in(\stacka_{\bmchanged{\stackb}}\cup\{[,]\})^*$ the set $\lab(\sstate,\stackc)$ of
all atoms $p$ such that $\stackc \in \lang(\wautop{\sstate})$. In
other words, one reads the higher-order stack in a bottom-up fashion to
determine which atoms are satisfied in the current configuration. We
refer the reader to \cite{CHMOS08} for related work on this notion of
regular sets of higher-order stacks (without links).

In the general case of collapsible pushdown compound Kripke structures,
regular labelling functions are defined using a richer model of
automata introduced in \cite[Section~3]{BCOS10}, that we recall
here. Note that if one considers stacks without links, this model
corresponds to the previous one.

Let $\stackc$ be an order-$k$ collapsible stack. We first associate
with $\stackc=\stackc_{1},\cdots,\stackc_{\ell}$ a well-bracketed word 
of bracket-depth $k$,
$\flaten{\stackc}\in(\stacka_\stackb\cup\{\lsk,\rsk\})^*$, defined as follows:
  \[ \flaten{\stackc} \; := \; \begin{cases}
                   \lsk\flaten{\stackc_{1}}\cdots\flaten{\stackc_{\ell}}\rsk & \text{if }k\geq 1\\
                   \stackc & \text{if }k=0 \text{ (\ie $\stackc\in\stacka_\stackb$)}\\
                 \end{cases}
  \]
  In order to reflect the link structure, we define a partial function
  $\target{\stackc}:\{1,\cdots ,|\flaten{\stackc}|\}\rightharpoonup \{1,\cdots
  ,|\flaten{\stackc}|\}$ that assigns to every position in
  $\{1,\cdots ,|\flaten{\stackc}|\}$ the index of the end of the stack
  targeted by the corresponding link (if it exists; indeed this is
{only defined if the symbol at this position is in
  $\Gamma\times\{2,\cdots,k\}\times\nat$}  
). Thus with $\stackc$ is associated
  the pair $\anglebra{\flaten{\stackc},\target{\stackc}}$; and with a set $\stackset$ of
  stacks is associated the set
  $\widetilde{\stackset}=\{\anglebra{\flaten{\stackc},\target{\stackc}}\mid \stackc\in \stackset\}$.

\begin{exa}
  Let
    \begin{center}\includegraphics[scale=.3]{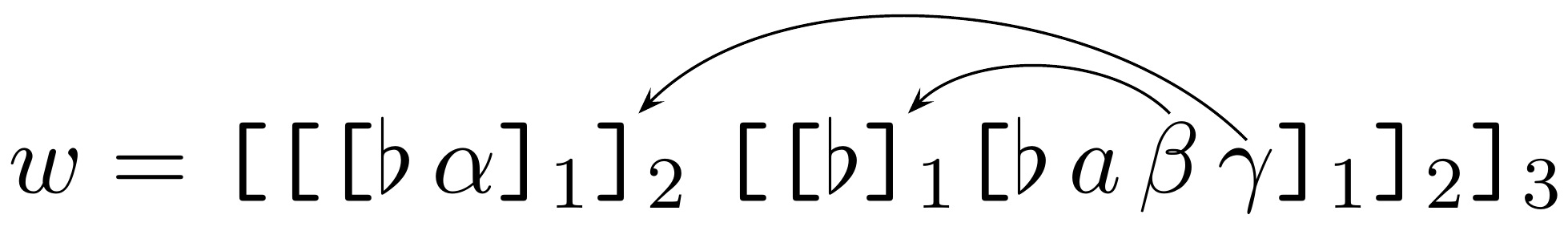}\end{center}

  Then 
  
    \[\flaten{\stackc} =\mksk{
 \mksk{\mksk{\bot \,\alpha}}\mksk{\mksk{\bot}\mksk{\bot\,\alpha\,\beta\,\gamma}}}
\] 
%
    {$\target{15}=7$, $\target{16}=11$, and $\target{i}$ is undefined for
  all other $i\in \{1,\ldots,|\flaten{\stackc}|=19\}$.}
\end{exa}

We consider \emph{deterministic} finite automata working on such
representations of collapsible stacks. The automaton reads the word
$\flaten{\stackc}$ from left to right. On reading a letter that does not
have a link (\ie~$\targetf$ is undefined on its index) the automaton
updates its state according to the current state and the letter; on
reading a letter that has a link, the automaton updates its state
according to the current state, the letter and the state it was in
after processing the targeted position. A run is accepting if it ends
in a final state.

Formally, such an automaton is a tuple
$\anglebra{Q,A,q_{\init},F,\delta}$ where $Q$ is a finite set of
states, $A$ is a finite input alphabet, $q_{\init}\in Q$ is the
initial state, $F\subseteq Q$ is a set of final states and
$\delta: (Q\times A) \cup (Q\times A\times Q) \rightarrow Q$ is a
transition function. With a pair $\anglebra{u,\tau}$ where
$u=a_1\cdots a_n\in A^*$ and $\tau$ is a partial map from
$\{1,\cdots n\}\rightharpoonup \{1,\cdots n\}$, we associate a unique run
$r=r_0\cdots r_n$ as follows:
\begin{itemize}
\item $r_0=q_{\init}$;
\item for all $0\leq i< n$, \[r_{i+1}=
  \begin{cases}
    \delta(r_i,a_{i+1}) & \text{if } i+1\notin \dom(\tau),\\
    \delta(r_i,a_{i+1},r_{\tau(i+1)}) & \text{otherwise}.
  \end{cases}\]
\end{itemize} 
The run is \emph{accepting} just if $r_{n}\in F$, and the pair $(u,\tau)$ is \emph{accepted} just if the associated run is accepting.

A \defin{regular labelling function} is given as a set of such
    automata $\wautop{\sstate}$ over alphabet $\stacka_{\bmchanged{\stackb}}\cup\{[,]\}$, one for each $p\in\APf$ and each
    $\sstate\in\setstates$. They define the labelling function that
    maps to each state $\sstate\in\setstates$ and stack with link $\stackc$
the set $\lab(\sstate,\stackc)$ of all atoms $p$ such that $\anglebra{\flaten{\stackc},\target{\stackc}}$ is accepted by $\wautop{\sstate}$.

\paragraph{Associated compound Kripke structure}

An order-$k$ \CPCKS $\CPCKS=(\stacka,\setstates,\relation,\lab,\sstate_\init)$
over local states $\{\setlstates_i\}_{i\in [n]}$  generates an infinite \CKS
$\CKS_\CPCKS=(\setstates',\relation',\lab',\sstate'_\init)$ over 
$\{\setlstates_i\}_{i\in [n+1]}$, where
\begin{itemize}
\item $\setlstates_{n+1}=\Stacks_k$,
\item $\setstates'=\setstates\times\Stacks_k$, 
\item $(\sstate',\stackc')\in\relation'(\sstate,\stackc)$ if 
$(\sstate,\stackc)\conftrans{}(\sstate',\stackc')$,
\item  $\lab'=\lab$ and
\item $\sstate'_\init=(\sstate_\init,\bot_k)$.
\end{itemize}

We write $\CPCKS\models\phi$ if $\CKS_\CPCKS\models\phiprime$, where $\phiprime$ is
obtained from $\phi$ by replacing each concrete observation $\cobs\subseteq
[n]$ with $\cobs'=\cobs\union\{n+1\}$, to reflect the fact that the
stack is always visible.

\subsection{Collapsible pushdown tree automata}\label{sec:CPTA}

\subsubsection{Definitions}\label{sec:CPTA-def}

We explain how to generalise to higher-order the definitions of
pushdown tree automata from Section~\ref{sec-direction-guided}. The
idea is simple: we now work with stacks with links instead of usual
stacks and the operation performed on the stack depends on the current
$\topone$ element. Formally, this leads to the following definition.

For $\APf$ a
  finite set of atomic propositions and $\Dirtree$ a finite set of
  directions, an \defin{order-$k$ collapsible alternating pushdown tree automaton
    (\CAPTA) on $(\APf,\Dirtree)$-trees} is a tuple
  $\auto=(\stacka,\tQ,\tdelta,\tq_{\init},\couleur)$ where $\stacka$
  is a finite stack alphabet with a special bottom symbol $\stackb$,
  $Q$ is a finite set of states, $\tq_{\init}\in \tQ$ is an initial
  state,
  $\tdelta : \tQ\times 2^{\APf}\times \stacka \rightarrow
  \boolp(\Dirtree\times \tQ\times\opn{k})$ is a transition function,
  and $\couleur:\tQ\to \setn$ is a colouring function. 

  Acceptance of a tree by a \CAPTA is again defined as a parity game,
  the only difference being that now the game we obtain is played on a
  richer underlying arena. While in the case of \APTA we had 
  pushdown games, we now obtain  collapsible pushdown games (see
  \cite{HMOS08} for more results on this). Note that such games are
  decidable, hence acceptance of a {regular} tree by a \CAPTA is decidable as
  well.

A \defin{collapsible nondeterministic pushdown tree
    automaton} (\CNPTA) is a collapsible alternating pushdown tree automaton
  $\nauto=(\stacka,\tQ,\tdelta,\tq_{\init},\couleur)$ such that for
  every $\tq\in \tQ$, $a\in 2^{\APf}$ and  $\stacks\in\stacka$,
  $\tdelta(\tq,a,\stacks)$ is written in disjunctive normal form and
  for every direction $\dir\in \Dirtree$, each disjunct contains
  exactly one element of $\{\dir\}\times Q\times \opn{k}$.

  The restrictions leading respectively to \emph{semi-alternating
    collapsible pushdown tree automata} and \emph{$X$-guided stack
    alternating collapsible pushdown tree automata} are essentially
  the same as in Section~\ref{subsec-direction-guided} except that now
  the requirement is that the stack operation is the same when going
  in the same direction (previously, we were requiring that the same
  content was pushed on the stack). Formally, we have the following
  definition (generalising Definition~\ref{def-guided}):

\begin{defi}
  
An order-$k$ \CAPTA $\auto=(\stacka,\tQ,\tdelta,\tq_{\init},\couleur)$ over $\Dirtree\times \Dirtreea$-trees has an
  \defin{$\Dirtree$-guided stack}, or simply is \defin{$\Dirtree$-guided}, if there exists a
  function $\tdeltastack: 2^{\APf}\times\stacka\times\Dirtree\to \opn{k}$ such that for all
  $(\tq,a,\stacks)\in\tQ\times 2^{\APf}\times \stacka$, all atoms
  appearing in
  $\tdelta(\tq,a,\stacks)$ are of the form $[(\dir,\dira),\tq',\tdeltastack(a,\stacks,x)]$.
\end{defi}

\subsubsection{Projection, simulation and narrowing}\label{sec:CPTA-res}

Following the same proof as for Proposition~\ref{theo-projection} we have the following generalisation to higher-order.

\begin{prop}
\label{theo-projection-HO}
  Given an \CNPTA $\NTA$  and  $p\in\APf$, one can build  an \CNPTA
  $\proj{\NTA}$
  such that for every pointed tree $(\ltree,\noeud)$ and initial {$k$-stack}
 {$\stackc_\init\in\Stacks_k$},
  \begin{align*}
    (\ltree{,\noeud})\in\,& \lang(\proj{\NTA},\stackc_\init) \mbox{\bigiff}\\
&    \exists \plab \mbox{ a $p$-\labeling for $\ltree$ s.t. }(\ltree\prodlab\plab,\noeud)\in\lang(\NTA,\stackc_\init).    
  \end{align*}
\end{prop}

Now, moving to simulation, one easily generalises the proof in~\cite{aminof2013pushdown} to higher-order. 

\begin{thm}
\label{theo-simulation-HO}
Given a semi-alternating \CAPTA $\SPTA$, one can build 
an \CNPTA $\NTA$ 
such that for every initial {$k$-stack $\stackc\in\Stacks_k$,} $\lang(\NPTA,\stackc)=\lang(\SPTA,\stackc)$. 
\end{thm}

\begin{proof}
The key idea in the proof in the pushdown case is to remark that it is sufficient to do a subset construction on the set of states as the stack is the same when moving down in the same direction in the tree. Here, the same approach is also working as the stack operation (hence the stack with links) is the same when moving down in the same direction in the tree.
\end{proof}

As the previous construction also preserves $\Dirtree$-guidedness,
we can refine the above result as follows:

\begin{prop}
\label{prop-stable-simulation-HO}
Given an $\Dirtree$-guided \CAPTA $\SPTA$, one can build an
$\Dirtree$-guided \CNPTA $\NPTA$ 
such that for every initial {$k$-stack $\stackc\in\Stacks_k$,} $\lang(\NPTA,\stackc)=\lang(\SPTA,\stackc)$.   
\end{prop}

Finally, narrowing directly extends to higher-order. 

\begin{thm}
Given a \CAPTA $\SPTA$ on $\Dirtree\times\Dirtreea$-trees, one can build 
a  \CAPTA ${\narrow[\Dirtree]{\APTA}}$ on $\Dirtree$-trees such
  that for every pointed $(\APf,\Dirtree)$-tree $(\ltree,\noeud)$,
  every $\noeud'\in(\Dirtree\times\Dirtreea)^+$ such that $\projI[\Dirtree]{\noeud'}=\noeud$,
  and every initial {$k$-stack $\stackc\in\Stacks_k$,} 
\[(\ltree,\noeud)\in\lang(\narrow[\Dirtree]{\APTA},\stackc_\init) \mbox{ iff
 }(\liftI[\Dirtreea]{}{\ltree},\noeud')\in\lang(\APTA,\stackc_\init).\]
\end{thm}

  \begin{prop}
  \label{prop-narrow-stable-HO}
    If a \CAPTA $\APTA$ over $\Dirtree\times \Dirtreea \times \Dirtreeb$-trees is $\Dirtree$-guided, then so
    is ${\narrow[\Dirtree\times\Dirtreea]{\APTA}}$.
  \end{prop}

\subsection{Model checking hierarchical \QCTLsi on collapsible pushdown compound Kripke structures}\label{sec-QCTL-HO}

We now describe how one establishes an extension to higher-order of Theorem~\ref{theo-decidable-QCTLi}. 

\begin{thm}
  \label{theo-decidable-QCTLi-HO}
Model checking \QCTLsih on collapsible pushdown compound Kripke structures is decidable.
\end{thm}

\subsubsection{Succinct unfoldings}

The first notion that needs to be generalised is the one of succinct
unfoldings. Recall that the idea was to consider trees over a
\emph{finite} set of directions, to later use tree automata. The trick
was to choose as set of directions for these trees the set of all
possible finite words that the pushdown system to be model-checked could push on the stack. Here,
as we have to handle stacks with links, the stack can be deeply
modified by a single transition and the set of all those possible
modifications is no longer finite. However, a simple solution consists
in choosing as set of directions the set of all possible
higher-order stack operations used by the collapsible pushdown
system.

Hence, it is enough to record only the operations made on the
stack in each node: then, starting from the root and the initial stack
content $\bot_k$, one can reconstruct the stack content at each node by
following the unique path from the root to this node, and applying the
successive operations on the stack. By doing so we obtain a tree over the finite set
of directions $\setstates\times\CDirstack$, where
\[\CDirstack=\{op\mid
  (\sstate,\stacks,\sstate',op)\in\relation \text{ for some
  }\sstate,\stacks \text{ and }\sstate'\}\] Note that we no longer
keep a neutral direction (the empty word in the pushdown setting): it
was previously useful in the proof of Lemma~\ref{lem-final} (subcase
$\bm{\phi=p}$) but in the higher-order setting we will need a more
involved tool as the simple trick of destroying the stack content
through direction $\epsilon$ will no longer be sufficient.

Definitions of the \emph{succinct representation of a partial path} as well as the \emph{succinct unfolding} (Definition~\ref{dec-succ-unfold}) are adapted to \CPCKS in the straightforward way (we keep the same notations). 

The following result is then proved as Lemma~\ref{lem-equivalence}.

\begin{lem}
  \label{lem-equivalence-HO}
  For every \CPCKS $\CPCKS$ over
 $\{\setlstates_i\}_{i\in [n]}$  
  and every \QCTLsi formula $\phi$, 
  $\CPCKS\models\phi$ \,iff\, $\unfold[\CPCKS]{\sstate_\init,\stackb}\models\phiprime$.
\end{lem}

\subsubsection{Proof of Theorem~\ref{theo-decidable-QCTLi-HO}}

The proof of Theorem~\ref{theo-decidable-QCTLi-HO} follows the same
lines as the one of Theorem~\ref{theo-decidable-QCTLi}, \ie it relies
deeply on an inductive construction of a tree automaton working on the
succinct representation (Lemma~\ref{lem-final} in the pushdown
setting).

In the higher-order setting this leads to the following statement (the
only difference is now that we replaced \APTA by \CAPTA and
consider a collapsible pushdown Kripke compound structure $\CPCKS=(\stacka,\setstates,\relation,\lab,\sstate_\init)$).

\begin{lem}
    \label{lem-final-HO}
    For every subformula
    $\phi$ of $\Phi_{n+1}$ and state $\sstate\in\setstates$,  one can build a
    \CAPTA $\bigauto[\sstate]{\phi}$ on
    $(\APq(\Phi),\dirphi)$-trees with $\CDirstack$-guided stack
and    such that for every 
    $(\APq(\Phi),\dirphi)$-tree $\ltree$ rooted in
    $(\projI[{\Iphi}]{\sstate_\init},\bot_k)$, every 
 partial path $\spath\in\PPaths(\CPCKS)$ ending in
    $\conf{\sstate}{\stackc}$, 
    it holds that
\[      (\ltree,\projI[{\Iphi}]{\succpath(\spath)})\in\lang(\bigauto[\sstate]{\phi},\stackc) \mbox{\;\;\;iff\;\;\;}
      \liftI[{\setstates\times\CDirstack}]{}{\ltree}\merge\;\unfold{\sstate,\stackc},\succpath(\spath) \modelst
      \phi. \]
  
  \end{lem}

\begin{proof}

As for Lemma~\ref{lem-final} the proof is by induction on $\phi$. The only case that differs from the pushdown case is the base case case where $\phi=p$ as it requires to handle regular labelling functions (which are now richer than in the pushdown case).

In the pushdown case, to check for a formula $\phi=p$ we followed the
dummy direction $(s,\epsilon)$ to destroy letter by letter the current
stack content while simulating on the fly $\wautop{\sstate}$.  In the
setting without links, a similar trick would work: one would read
letter by letter the higher order stack performing $\popone$
operations, or $\popn{i+1}$ operations when getting to an empty
topmost $i$-stack which can be handled thanks to a small change of the
model where one can test whether the topmost-$i$ stack is empty (this
can be simulated by the present model); see
e.g.~\cite{CarayolPHD,FrataniPHD}.  However for stacks with links
this approach no longer works as one also needs to follow the links
and this would require to destroy the stack.

The solution in the general case of \CAPTA is to anticipate these
tests and to enrich the automaton so that it has in its control states
an extra component that, for every $\wautop{\sstate}$, gives the state
reached in $\wautop{\sstate}$ after processing the current stack
content. In the pushdown automaton it is an easy exercise how to
compute such an enriched version ({when pushing some content
  one simply simulates $\wautop{\sstate}$ on the new symbols
  added, and pushes  each symbol together with  its corresponding state of
  $\wautop{\sstate}$; popping is then for free}). In the case of
  higher-order pushdown~\cite{CHMOS08} and collapsible
  pushdown~\cite{BCOS10} it is a highly non-trivial result (we
  rephrase it here for \CAPTA but the proof ingredients are the same).

\begin{thm}\cite[Theorem~3]{BCOS10}
\label{thm-closure-CAPTA}
Given an order-$k$ \CAPTA $\mathcal A$ {with a state-set $Q$} and an
automaton $\mathcal{B}$ (that takes as input stacks with links over the same
alphabet as $\mathcal{A}$), there exist an order-$k$ \CAPTA
${\mathcal A}[\mathcal{B}]$ with  state-set $Q'$, a subset
${F} \subseteq Q'$ and a mapping {$\chi : Q' \rightarrow Q$} such
that:
\begin{itemize}
\item[(i)] ${\mathcal{A}}$ and
  ${\mathcal{A}[\mathcal{B}]}$ accept the same trees. \oschanged{}
\item[(ii)] for every configuration $\conf{q}{\stackc}$ of
  {$\mathcal{A}[\mathcal{B}]$}, the corresponding
  configuration\oschanged{\footnote{More precisely,
      ${\mathcal{A}[\mathcal{B}]}$ works on a stack alphabet and set
      of control states that extend those of $\mathcal{A}$, and
      configurations of $\mathcal{A}$ are obtained from configurations
      of ${\mathcal{A}[\mathcal{B}]}$ by forgetting the extra
      components from the control state and from the stack
      symbols. Hence, one should think of ${\mathcal{A}[\mathcal{B}]}$
      as a version of $\mathcal{A}$ with extra information stored
      both in the control states and in the stack symbols, and this
      information is precisely used to check whether the current
      stack content is accepted by $\mathcal{B}$.}} of $\mathcal{A}$
  has state $\chi(q)$ and its stack content is accepted by
  $\mathcal{B}$ if and only if $q \in {F}$.
\end{itemize}
\end{thm}

Now applying the construction from Theorem~\ref{thm-closure-CAPTA} at
every step in the inductive construction gives the base case $\phi=p$
for free.

The rest of the proof is similar to the one of Lemma~\ref{lem-final}.
\end{proof}

\subsection{Model checking hierarchical instances of \SLi on collapsible pushdown games arenas with visible stacks}\label{sec-SLi-HO}

Regarding \SLi we need to generalise the notion of pushdown game arena with visible stack to higher-order. 

\begin{defi}
\label{def-pgs-HO}
An order-$k$ \defin{Collapsible Pushdown Game Arena with Visible Stack}, or \CPGA for short, is a
tuple $\CPGA=(\Act,\stacka,\cstates,\ptrans,\val,\cstate_\init,\obsint)$ where
   \begin{itemize}
    \item $\Act$ is a finite set of actions,
    \item $\stacka$ is a finite stack alphabet together with a  bottom
      symbol $\stackb\notin\stacka$ and we let $\stacka_\stackb=\stacka\cup\{\stackb\}$,
    \item $\cstates$ is a finite set of control states,
    \item $\ptrans:\cstates\times\stacka_\stackb\times \Mov^{\Agf}\to \cstates\times{\opn{k}}$ is a transition function, 
    \item $\val:\cstates\times\Stacks_k\to 2^{\APf}$ is a regular labelling function (as defined in Section~\ref{sec-CKS-HO}), 
    \item $\cstate_\init \in \cstates$ is an initial control state, and
    \item  $\obsint:\Obsf\to 2^{\cstates\times\cstates}$ is an observation interpretation.
  \end{itemize}  
\end{defi}

A \defin{configuration} is a pair
$\conf{\cstate}{\stackc}$ where $\cstate\in\cstates$
represents the current control state and $\stackc$ is the current content
of the stack with links.  When the players choose a joint move
$\jmov\in \Mov^{\Agf}$ in a configuration
$\conf{\cstate}{\stackc}$ the system moves to
configuration $\conf{\cstate'}{op(\stackc)}$, where
$\conf{\cstate'}{op}=\ptrans(\cstate,\topone(\stackc),\jmov)$; we denote
this by
$\conf{\cstate}{\stackc}\conftrans{\jmov}\conf{\cstate'}{op(\stackc)}$.

A \CPGA
$\CPGA=(\Act,\stacka,\cstates,\ptrans,\val,\cstate_\init,\obsint)$
induces an infinite \CGS
$\CGS_\CPGA=(\Act,\vertices',\trans,\val',\pos_\init,\obsint')$ where
\begin{itemize}
\item $\vertices'=\cstates\times \Stacks_k$,
\item $\trans(\conf{\cstate}{\stacks\cdot\stackc},\jmov)=
\conf{\cstate'}{\stackc'\cdot\stackc}$ if
$\conf{\cstate}{\stacks\cdot\stackc}\conftrans{\jmov}\conf{\cstate'}{\stackc'\cdot\stackc}$,
\item $\val'=\val$,
\item $\pos'_\init=\conf{\cstate_\init}{[\bot_k]}$,
\item
  $(\conf{\cstate}{\stackc},\conf{\cstate'}{\stackc'})\in\obsint'(\obs)$
  if $\stackc =
\stackc' \mbox{ and }(\cstate,\cstate')\in\obsint(\obs)$.
\end{itemize}
{We call plays and partial plays of $\CPGA$ those of $\CGS_\CPGA$.}
For an \SLi sentence $\phi$, we write $\CPGA\models\phi$ if $\CGS_\CPGA\models\phi$.

Decidability of hierarchical instances of \SLi on collapsible pushdown
game arenas with visible stack follows the same line as in the
pushdown setting. Indeed, the reduction to \QCTLsi works the same (the
only point to check is that the labellings obtained in the reduction
are regular ones but this is immediate).

\begin{thm}
\label{theo-SLi-HO}
The model-checking problem for \SLi on collapsible pushdown game arenas with
visible stack is decidable for hierarchical instances.
\end{thm}

\section{Conclusion}
\label{sec-conclusion}

We proved that we can model check Strategy Logic with imperfect
information on collapsible pushdown game arenas when the stack is
visible and information  hierarchical. {This implies that, on such
infinite systems and for LTL objectives, one can decide the existence
of Nash equilibria or solve a variety of 
 synthesis problems such as distributed synthesis, rational synthesis 
or assume-guarantee synthesis, all easily
expressible in Strategy Logic.}

{Strategy Logic is also known to be decidable, with elementary
complexity, on imperfect-information arenas where actions are public~\cite{belardinelli2017verification}. One interesting
future work would be to extend also this result to the pushdown  setting.}

Another possible continuation of this work is to consider synthesis rather than model-checking. More specifically, to start with a collapsible pushdown system with controllable and uncontrollable actions and a specification in Strategy Logic with imperfect information, and ask for a controller that restricts the system so that the specification is satisfied. Of course, due to the non-elementary underlying complexity of the model-checking problem, one should restrict to sub-classes (of system and/or formulas) to hope for tractable results.

\newpage
\bibliographystyle{alpha}
\newcommand{\etalchar}[1]{$^{#1}$}

\end{document}